\documentclass[11pt]{article}
\usepackage{amsmath,amssymb,amsfonts} 
\usepackage{epsfig} \usepackage{latexsym,nicefrac,bbm}
\usepackage{xspace}
\usepackage{color,fancybox,graphicx,subfigure,fullpage}
\usepackage[top=1.25in, bottom=1.25in, left=1in, right=1in]{geometry}
\usepackage{tabularx} \usepackage{hyperref} 
\usepackage{pdfsync}
\usepackage[boxruled]{algorithm2e}
\usepackage{multicol}
\usepackage{enumitem}
\newcommand{\parag}[1]{ {\bf \noindent #1}}

\newcommand{\nfrac}{\nicefrac}

\newcommand{\supp}{\mathrm{supp}}

\newcommand{\prob}[1]{\mathrm{Pr}\insquare{#1}}

\newcommand{\conv}{\mathrm{conv}}
\newcommand{\dist}{\mathrm{dist}}
\newcommand{\argmin}{\operatornamewithlimits{argmin}}

\newcommand{\fc}{\mathrm{fc}}

\newcommand{\cB}{\mathcal{B}}

\newcommand{\cF}{\mathcal{F}}
\newcommand{\capa}{\mathrm{Cap}}

\newcommand{\st}{\mathrm{s.t.}}

\newcommand{\wt}{\widetilde}
\newcommand{\diam}{\mathrm{diam}}
\newcommand{\lspan}{\mathrm{span}}
\newcommand{\interior}{\mathrm{int}}

\newcommand{\per}{\mathrm{per}}
\newcommand{\bl}{\mathrm{BL}}

\newcommand{\eps}{\varepsilon}

\newenvironment{proof}{\noindent{\bf Proof:}\hspace*{1em}}{\qed\bigskip}
\newcommand{\qed}{\hfill\ensuremath{\square}}

\newcommand\Z{\mathbb Z}
\newcommand\N{\mathbb N}
\newcommand\R{\mathbb R}
\newcommand\C{\mathbb C}

\newtheorem{theorem}{Theorem}[section]
\newtheorem{fact}{Fact}[section]

\newtheorem{definition}{Definition}[section]
\newtheorem{lemma}[theorem]{Lemma}
\newtheorem{remark}[theorem]{Remark}

\newtheorem{corollary}{Corollary}[section]

\def\abs#1{\left| #1 \right|}
\newcommand{\norm}[1]{\ensuremath{\left\lVert #1 \right\rVert}}

\newcommand{\diag}[1]{{\sf Diag}\left({#1}\right)}

\newcommand{\poly}{\mathrm{poly}}

\newcommand{\inparen}[1]{\left(#1\right)}             
\newcommand{\inbraces}[1]{\left\{#1\right\}}           
\newcommand{\insquare}[1]{\left[#1\right]}             
\newcommand{\inangle}[1]{\left\langle#1\right\rangle} 

\newenvironment{proofof}[1]{\begin{trivlist} \item {\bf Proof
#1:~~}}
  {\qed\end{trivlist}}

\title{\bf Maximum Entropy Distributions: \\ Bit Complexity and Stability}

\usepackage{authblk}

\author[1]{ Damian Straszak}
\author[2]{Nisheeth K. Vishnoi}
\affil[1]{\small Aleph Zero Foundation}
\affil[2]{\small Yale University}

\begin{document}

\maketitle
 
\begin{abstract}
Maximum entropy distributions with discrete support in $m$ dimensions  arise in machine learning, statistics, information theory, and theoretical computer science.
While structural and computational properties of max-entropy distributions have been extensively studied,  basic questions such as: {\em Do max-entropy distributions over a large support  (e.g., $2^m$) with a specified marginal vector have succinct descriptions (polynomial-size in the input description)?} and: {\em Are entropy maximizing distributions ``stable'' under the perturbation of the marginal vector?}
have resisted a rigorous resolution. 
Here we show that these questions are related and resolve both of them.
Our main result shows  a ${\rm poly}(m, \log 1/\eps)$ bound on the bit complexity of $\eps$-optimal dual solutions to the maximum entropy convex program -- { for very general support sets and with no restriction on the marginal vector.}
Applications of this result include polynomial time algorithms to compute  max-entropy distributions over several new and old polytopes for any marginal vector in a unified manner, a  polynomial time algorithm to compute the Brascamp-Lieb constant in the rank-1 case.
The proof of this result allows us to show that changing the marginal vector by $\delta$ changes the max-entropy distribution in the total variation distance roughly by a factor of ${\rm poly}(m, \log 1/\delta)\sqrt{\delta}$ -- { even when the size of the support set is exponential.} 
Together, our results put max-entropy distributions on a mathematically sound footing --  these distributions are  robust and computationally feasible models for data.
  
\end{abstract}

\newpage

\tableofcontents

\newpage

\section{Introduction}
 
The central objects of study in this paper are maximum entropy (max-entropy) distributions over discrete domains.
For a  finite set of vectors $\cF \subseteq \Z^m$ and a point $\theta \in \R^m$, the max-entropy distribution corresponding to  $\theta$ is defined to be the one that maximizes the entropy  among all distributions  $\{q_\alpha\}_{\alpha \in \cF}$ whose expectation is $\theta$.
This is  formulated as the following convex program:\footnote{For simplicity, in the introduction we only consider the case with a  uniform prior.} 
\begin{equation}
\textstyle \begin{aligned}
\textstyle	\sup ~~ & \sum_{\alpha \in \cF}q_\alpha \log \frac{1}{q_\alpha},&\\
		\st ~~ & \sum_{\alpha \in \cF} q_\alpha \cdot \alpha = \theta, \;   \;
 \sum_{\alpha \in \cF} q_\alpha=1,& \; \;   q\geq 0. 
\end{aligned}
\label{eq:max_entropy_informal}	
\end{equation}
The model of a max-entropy distribution is natural to consider when one has  an unknown distribution $q$ on the set $\cF$ and the only information available is its marginal vector  $\theta.$ 
Then, according to the max-entropy principle~(\cite{Jaynes1,Jaynes2}), the best guess for $q$ is the max-entropy distribution with expectation $\theta$. 
The rationale behind this choice is that maximizing the entropy yields a distribution with the least amount of prior information built-in.

In many interesting cases, the support $\cF$ is a large set and may be specified only implicitly (e.g., as the set of all subsets of a universe or the set of all spanning trees or matchings in a graph).
Hence, the program~\eqref{eq:max_entropy_informal} may have a prohibitively large number of variables, $|\cF|$, that makes solving it	intractable.
Thus, one often considers the following Lagrange dual of~\eqref{eq:max_entropy_informal} that involves only $m$ variables.	
\begin{equation}\label{eq:dual_simplified}
\begin{aligned}{\textstyle \inf_{y\in \R^m}  \log \inparen{\sum_{\alpha \in \cF} e^{\inangle{\alpha - \theta,y}}}.}
\end{aligned}
\end{equation}

\noindent
Two questions arise: is  max-entropy a ``stable'' model? In other words, how much does the optimal distribution (in total variation distribution) change when one changes the marginal vector? And, given the marginal vector, can one compute the max-entropy distribution?
Both these questions are important from the point of view of applicability of the max-entropy model.
Stability is an important consideration in machine learning (\cite{bousquet2002stability,yu2013,HuangV19}): in particular, for the max-entropy problem, when the marginal vector comes from averaging ``observed'' samples, stability plays a role in deciding the number of samples to learn the distribution; see \cite{Dudik07}.
Computability of max-entropy distributions enables one to estimate the max-entropy distribution corresponding to the above marginal vector and has been studied extensively in theoretical computer science (\cite{SinghV14}) and machine learning (\cite{Dudik07}).  

\paragraph{Computability and bit complexity.}
Duality implies that the optimal solution $q^\star$ to~\eqref{eq:max_entropy_informal} is succinctly represented by the optimal dual solution~\eqref{eq:dual_simplified} as follows:
$q_\alpha^\star \propto e^{\inangle{\alpha, y^\star}}$  {for every } $\alpha \in \cF$.
While the above provides a way to represent the max-entropy distribution $q^\star$ using only a vector of $m$ real numbers, the issue becomes the bit complexity of the optimal solution $y^\star$. 
Indeed, any algorithm that tries to compute $y^\star$ requires the knowledge of an upper bound $\norm{y^\star}\leq R$.
Further, for the representation above to be efficient one requires $\norm{y^\star}$ to be bounded by a polynomial in $m$, as (even assuming $\norm{\alpha}=O(1)$ for all $\alpha\in \cF$) the value $e^{\inangle{\alpha,y^\star}}$ can be as large as $e^{\Omega(\norm{y^\star})}$ and thus requires $\Omega(\norm{y^\star})$ bits to store.

The first step to  understand the bit complexity of max-entropy distributions for general polytopes  was taken in~\cite{SinghV14} where it was proved that whenever the point $\theta$ is ``deep'' inside the marginal polytope $P:=\conv\{\alpha: \alpha \in \cF\}$, then a polynomial bound on $y^\star$ holds.
More precisely, whenever $\theta\in P$ is at a  distance $\eta>0$ away from the boundary of $P$ then $\norm{y^\star}\leq \frac{\log |\mathcal{F}|}{\eta}$.

However, in many  applications one has to deal with instances where the point $\theta$ lies very close or even on the boundary of $P$.
For instance, in machine learning and statistics~(\cite{PietraPL97,Nigam99,Soofi00,SW87,Dudik07}), the point $\theta$ is often obtained as the empirical average of observed samples. 
Here,  it might happen that the point $\theta$ lies in the interior with only exponentially small probability or may end up on the boundary of $P$ even if $P$ is rather ``round''; see Section~\ref{eq:sec_boundary}. 
In TCS, several problems  can also be seen to be  disguised forms of the max-entropy problem and polynomial time algorithms for them necessarily require  a polynomial bound on the size of the optimal or near-optimal dual solution.
Examples include the problems of computing Brascamp-Lieb constants (\cite{GGOW17}) in the rank-1 setting,
deterministically approximating mixed discriminants (\cite{GS02}),
 solving convex relaxations for counting and optimization problems  (\cite{NS16}).

In all of the above applications, the bound given by~\cite{SinghV14} is not  sufficient as it necessarily deteriorates to $+\infty$ when $\eta$ tends to $0$ and one needs the following stronger result:
{\em For a given $P$, if $y^\star_\eps$  is an $\eps$-approximate solution  to the dual problem, then $\|y_\eps^\star\| \leq R_\eps= \poly\inparen{m,\log \nfrac{1}{\eps}}$ -- where the bound holds  for {\bf every} $\theta \in P$.}
Our first result shows that the above is true as long as all the faces describing $P$ have integer coefficients that are polynomially bounded by $m$. 
This  includes all the aforementioned examples as special cases and also covers a large class of combinatorial polytopes (e.g., matroid polytopes, matroid intersection polytopes such as bipartite matching, general matching, submodular polytopes, Brascamp-Lieb polytopes).
This result is obtained by a  geometrical reasoning that involves convex analysis and polyhedral geometry: it  relies on the fact that the faces describing the convex hull of the discrete domain have polynomially small coefficients, and 
avoids combinatorial arguments based on the specific structure of the domain.
We also complement this result by providing a matching bit-complexity lower-bound for polytopes whose faces have large coefficients.
Surprisingly, we observe that the bit complexity of max-entropy distributions also is related to its stability. 

\paragraph{Stability.} 
Using the techniques in the proof of our bit complexity result,  we  show that max-entropy distributions are ``stable'': If we consider two marginal vectors $\theta_1 $ and $\theta_2$, then the total-variation distance between the arising max-entropy distributions is bounded by $\poly\inparen{m,\log \nfrac{1}{\eps}} \cdot \sqrt{\|\theta_1-\theta_2\|}.$
Note that there is no a priori reason for the stability bound to be $\poly\inparen{m,\log \nfrac{1}{\eps}}$ -- it could be as large as  $|\mathcal{F}|$ that could be exponential in $m$.
While stability of max-entropy distributions is  assumed to hold in practice, and is crucial for making the model meaningful, ours is the first result to formally establish polynomial stability for exponentially-sized support sets under   fairly general conditions.

\section{Our results}\label{ssec:results}

We provide simplified and slightly informal variants of our theorems  here.

\vspace{-3mm}
\subsection{Bit complexity of max-entropy distributions}
\vspace{-2mm}

We assume  that the marginal polytope has a low ``unary facet complexity'', as formalized below\footnote{{Note that this condition is different from the {\em binary facet complexity} defined in \cite{Grotschel1988}.}}. 
\begin{definition}[Unary facet complexity -- Informal; see Definition~\ref{def:fc}]\label{def:fc_informal}
Let $P \subseteq \R^m$ be a  polytope with integer vertices. The {\bf unary facet complexity} of $P$ ($\fc(P)$) is the smallest number $M \in \N$, such that $P$ can be described by linear inequalities with coefficients in  $\{-M,\ldots,0,\ldots, M\}$.
\end{definition}

\noindent
The class of polytopes with $\fc(P)=O(1)$ is very general and includes most polytopes of interest including matroid polytopes, matroid intersection polytopes,  matching polytope, submodular polytopes and  Newton polytopes of certain classes of polynomials (such as real stable polynomials).

The first result of this paper is  the below theorem, where $h_{\theta}(y)$ denotes the objective function in the dual program~\eqref{eq:dual_simplified} and $g(\theta)\in \R$ denotes the optimal value; i.e., $g(\theta):=\inf_{y\in \R^m} h_{\theta}(y)$.

\vspace{-2mm}
\begin{theorem}[Bit complexity of max-entropy - Simplified; see Theorem~\ref{thm:bound}]\label{thm:bound_informal}
There exists a bound $R_\eps=\mathrm{poly}\inparen{m, \log \frac{1}{\eps}}\in \R_{>0}$,   such that for every set $\cF\subseteq \Z^m$ with unary facet complexity of $\conv(\cF)$ polynomially bounded\footnote{Specifically, we require the bound on the unary complexity to be $m^c$ for some constant $c>0$ fixed in advance.}, for every $\eps>0$ we have
$$\forall ~ {\theta\in \conv(\cF)},~~ \exists ~{y\in B(0,R)},~~\mathrm{s.t.}~~ h_{\theta}(y) \leq g(\theta) + \eps.$$
\end{theorem}

\noindent
Here $B(0,R)$ denotes the unit $\ell_2$-ball of radius $R$ around the origin.
Note that the bound we provide
is {\em uniform} over the whole polytope $P$, meaning that the radius $R_\eps$ does not depend on $\theta$
-- this is crucial for some of the applications.
 The error of $\eps$ that we introduce when going from the optimal solution $y^\star$ to a near-optimal solution $y$ (in Theorem~\ref{thm:bound_informal}) translates to an $\ell_1$ distance of at most $\sqrt{\eps}$ between the max-entropy distribution $q^\star$ and the distribution $q$ obtained from the near-optimal dual solution $y$; see Corollary~\ref{cor:rep}.
This means that we do not introduce a large error when using approximate dual solutions and, importantly, we have a good control over this error as the bound on $R_\eps$ in Theorem~\ref{thm:bound_informal} depends on $\eps$ logarithmically.

Theorem~\ref{thm:bound_informal} implies polynomial time computability of maximum entropy distributions for polytopes with polynomial unary facet complexity and for all points $\theta \in \conv({\cF})$ (assuming existence of a suitable counting oracle) -- a question that was left open in the work of \cite{SinghV14}.\footnote{A variant of Theorem~\ref{thm:bound_informal} which applies to families $\cF \subseteq \N^m$ that are spanning trees of undirected graphs was obtained in~\cite{AGMGS10} and, for matroids (and more generally -- jump systems) in~\cite{AO17}.
However, these arguments  heavily rely on the combinatorial fact that such families admit greedy optimization algorithms and do not seem to generalize beyond this setting.} 
For a formal statement of this result and a proof, we refer to Section~\ref{sec:proof_comp}.

This result has both practical and theoretical consequences.
In practical situations, when the size of the discrete domain is large enough to force one to use the dual formulation of the max-entropy program, our bound may give a fast algorithm with provable guarantees. 
We show in Section \ref{eq:sec_boundary} that if the vector $\theta$ is an empirical mean then, even in the simplest of settings, it is exponentially close to the boundary of the marginal polytope justifying the need for our result. 
Theoretically, this result  helps derive a few  algorithmic results in a unified manner. For applications of Theorem~\ref{thm:bound_informal} to  the matrix scaling problem, computing Brascamp-Lieb constants, and computing polynomial-based relaxations for counting problems, we refer to Section~\ref{sec:applications}.
We also prove the following theorem that complements Theorem~\ref{thm:bound_informal}.

\vspace{-2mm}
\begin{theorem}[Lower bound]\label{thm:lower_bound}
For every bound $R=R\inparen{m, \log \frac{1}{\eps}}\in \R_{>0}$, polynomial in $m$ and $\log \frac{1}{\eps}$
there exists an $m \in \N$, an $\eps>0$ and a set $\cF \subseteq \Z^m$ with $\norm{\alpha} \leq O(m^{3/2})$ for every $\alpha \in \cF$ such that
$$\exists~{\theta\in \conv(\cF)}~~ \forall~{y\in B(0,R)}~~\mathrm{s.t.}~~ h_{\theta}(y) > g(\theta) + \eps.$$
\end{theorem}
\noindent 
The above result relies on existence of ``flat'' $0-1$ polytopes by Alon and Vu \cite{AV97}, i.e., full-dimensional polytopes whose vertices are in $\{0,1\}^m$ and which fit between two $e^{-\Theta(m \log m)}$-close hyperplanes.
Its proof appears in Section \ref{sec:proof_lower}.
We note that the lower bound  in Theorem~\ref{thm:lower_bound} is based on a ``bad example'' $\cF \subseteq \Z^m$ with $\norm{\alpha} =O(m^{3/2})$ for every $\alpha\in \cF$. It is an interesting open question whether similar examples can be constructed in the $0-1$ regime, i.e., when $\cF \subseteq \{0,1\}^m$.

\vspace{-2mm}

\subsection{Stability of max-entropy distributions}

\begin{theorem}[Stability of  max-entropy distrib. - Simplified; see Theorem~\ref{thm:stability_gen}]\label{thm:stability}
There exists a bound $R_\eps=\mathrm{poly}\inparen{m, \log \frac{1}{\eps}}\in \R_{>0}$,      such that for every set $\cF\subseteq \Z^m$ with $\mathrm{fc}(\conv(\cF))$ is polynomially bounded and for every $\theta_1, \theta_2 \in P$, if $\norm{\theta_1 - \theta_2}_2 \leq \eps$ then 
$\norm{q^{\theta_1} - q^{\theta_2}}_1\leq R_{\eps} \sqrt{\eps}.$

\end{theorem}

\noindent
In machine learning and statistics applications,  $\theta$ is typically computed as an average of observed samples and our stability result implies that even if there is inverse polynomial (in $m$) variance in this empirical marginal, we can still recover a distribution close to the max-entropy distribution.
Thus, this result shows in a rigorous manner and in very general, exponential-sized domains, that max-entropy distributions provide  a stable model for density estimation.
The proof of Theorem~\ref{thm:stability} crucially relies on the techniques that go in the proof of Theorem~\ref{thm:bound_informal} and is presented in Section \ref{sec:stability}.
It would be  interesting   to improve $\sqrt{\eps}$ to $\eps$ in our stability bound.

\vspace{-3mm}
\subsection{Computability of Brascamp-Lieb constants in the rank-1 regime}
As a corollary of our bit complexity result we present an application to computing Brascamp-Lieb constants in the rank-1 regime.
Brascamp-Lieb inequalities are an ultimate generalization of many inequalities used in analysis and all of mathematics, such as the H\"older inequality and Loomis-Whitney (\cite{Brascamp02, Lieb90}). 
Recently, these inequalities have been studied from the computational point of view~(\cite{GGOW17}), where the main problem is to compute quantities of the following form -- called Brascamp-Lieb constants.
Given a collection of matrices $B=(B_1, B_2, \ldots, B_m)$ with $B_j \in \R^{n_j \times n}$ and a point $p\in \R^m_{\geq 0}$  compute
\begin{equation}\label{eq:bl_const}
\textstyle \bl(B,p):=\inf\inbraces{\frac{\det\inparen{\sum_{j=1}^m p_j B_j^\top X_j B_j}}{\prod_{j=1}^m \det(X_j)^{p_j}}: X_j\in \R^{n_j \times n_j}, X_j\succeq 0, j=1,2, \ldots, m}.
\end{equation}
\noindent
The constant $\bl(B,p)$ is non-zero whenever $p$ belongs to the so-called Brascamp-Lieb polytope $P_B \subseteq \R^m$:
$$\textstyle P_B := \inbraces{p\in \R^m_{\geq 0}: \sum_{j=1}^m p_j \dim(B_j U) \geq \dim(U), \mbox{ for every lin. subspace }U\subseteq \R^n}.$$

\noindent 
Recently~\cite{GGOW17} gave a method for calculating the Brascamp-Lieb constant in polynomial time when the vector $p$ is rational and given in unary. 
Note that this does not imply a polynomial time algorithm in the classical sense (when the vector $p$ is given in binary).
Here we consider the special, but already non-trivial, case  when the matrices are of rank $1$; i.e., $B_j\in \R^{1\times n}$ for every $j=1,2, \ldots, m$. 
By interpreting Brascamp-Lieb constants in the rank-$1$ regime as solutions to certain entropy-maximization problems we can deduce from Theorem~\ref{thm:bound_informal}  that they can be calculated, up to a multiplicative precision $\eps>0$, in time polynomial in $m$ and $\log \frac{1}{\eps}$.

Our entropy interpretation also leads to an algorithm for computing worst-case Brascamp-Lieb constants over the whole Brascamp-Lieb polytope, i.e., $C(B):=\sup_{p\in P_B}\bl(B,p)$,  in the rank-1 regime.
This quantity can be used  as a universal constant (for any $p\in P_B$) for the  reverse Brascamp-Lieb inequality~(\cite{Barthe98}). 

\begin{theorem}[Computability of BL-constants]\label{thm:bl_comp}
Consider a sequence of $m$ real-valued matrices $B_1, B_2, \ldots, B_m \in \R^{1 \times n}$. There is an algorithm that
\begin{enumerate}
\item computes a multiplicative $(1+\eps)$-approximation of the the Brascamp-Lieb constant $\bl(B,p)$ in time polynomial in the description size of $(B_1, B_2, \ldots, B_m, p)$, $m$ and $\log \frac{1}{\eps}$,
\item computes a multiplicative $(1+\eps)$-approximation of the worst-case Brascamp-Lieb constant $\sup_{p\in P_B} \bl(B,p)$ in time polynomial in the description size of $(B_1, B_2, \ldots, B_m)$, $m$ and $\log \frac{1}{\eps}$.
\end{enumerate}
\end{theorem}
\noindent
For the proof, we refer the reader to Section~\ref{sec:app_bl}.

\section{Technical overview}\label{sec:overview}
\parag{Theorem~\ref{thm:bound_informal} -- previous work and possible approaches. }
In the context of Theorem~\ref{thm:bound_informal}, the result of \cite{SinghV14} can be restated as follows: whenever a ball of radius $\eta>0$ centered at $\theta$ is contained in $P:=\conv(\cF)$   the optimal solution $y^\star$ of~\eqref{eq:dual_simplified} exists and its length $\norm{y^\star}$ is bounded by $O\inparen{\frac{\log| \cF| }{\eta}}$.
The proof is simple:  the $\frac{1}{\eta}$ term comes from the fact that after scaling by $\log |\cF|$, the optimal solution $y^\star$ of the dual program belongs to the polar of the radius-$\eta$ ball around $\theta$. 
(The polar of radius-$\eta$ ball is itself a ball of radius $1/\eta$.) 

One could consider the following two  approaches to prove Theorem~\ref{thm:bound_informal} that try to take advantage of the above mentioned bound: {\em centering} and {\em projection}.
Both are based on slightly moving $\theta$ to a new point $\theta'$ so that the result by~\cite{SinghV14} is applicable to $\theta'$ and then reasoning that the small shift does not affect the optimal dual solution.
Centering is based on moving the point more towards the interior, e.g., by taking the centroid of the polytope $P$ (which is far from the boundary) and taking a small step from $\theta$ towards $\theta'$.
One can then prove that $\theta'$ is well in the interior of the polytope, and hence a suitable bound for it follows.
The second idea is based on projections: start with any point $\theta$, if $\theta$ is already (inverse-polynomially) far   from the boundary of $P$, then the result of~\cite{SinghV14} implies a suitable bound.
Otherwise, project the point $\theta$ onto the closest facet of $P$ and  recurse. 
By doing this, either the resulting new point $\theta'$ will end up in a vertex of $P$, or on a lower-dimensional face of $P$, where it is far  from the boundary. 

In both approaches, proving the result for the new point $\theta'$ is easy, given~\cite{SinghV14}; what remains is to  bound  the error introduced when moving from $\theta$ to $\theta'$.
Proving such a bound on the error turns out to be a non-trivial task. 
In the case of centering one has to pick an appropriate direction along which the point $\theta$ is moved.
Similarly, the first challenge in the projection approach is to even define a suitable ``projection operator'' on a facet which would behave as expected and do not cause the point $\theta$ to land outside of the polytope $P$ (as  Euclidean projections might do).
An issue which concerns both approaches is to deal with the behavior of the function $g$ close to the boundary, where it can be shown to be non-Lipschitz\footnote{In the simple case when $\cF:=\{0,1\}$, the function $g:[0,1] \to \R$ is of the form $g(\theta)=-\theta \log \theta-(1-\theta)\log (1-\theta)$. 
One can see that on every interval $(\eps, 1-\eps)$ for $\eps>0$ the function $g$ is Lipschitz, but it is not Lipschitz on $(0,1)$.} and hence very susceptible to local perturbations.
More precisely: moving from $\theta$ to $\theta'$ guarantees that the gradient $\nabla_y h_{\theta}(y)$ at $y=y'^\star$ -- the dual optimal solution at $\theta'$ -- is small, however, it is a challenge to derive a guarantee on the gap $g(\theta) - h_{\theta}(y'^\star)$ as the function $y\mapsto h_{\theta}(y)$ is not strongly convex\footnote{One can again consider the case of $\cF:=\{0,1\}$ -- the function $h(0,y)$ is then of the form $h(0,y)=\log\inparen{1+e^y}$. Note that $h(0,y)$ is a convex function of $y$, but $\frac{d^2}{d y^2} h(0,y) \to 0$ whenever $y \to \pm \infty$, hence the function is essentially ``flat'' at infinity, and not strongly convex.}.
In fact, even a weaker claim that $g(\theta)\approx g(\theta')$ does not follow from~\cite{SinghV14}, since the convexity argument (as in Lemma 7.5 in the full version of ~\cite{SinghV14}) only allows one to prove that $g(\theta')\geq g(\theta) - O(\eps)$, where $\eps$ denotes the distance between $\theta$ and its ``centering'' $\theta'$, but not in the opposite direction, as required.

Our proof of Theorem~\ref{thm:bound_informal} is based on a purely geometric reasoning and bypasses the above obstacles by working entirely in the dual space (working in which we believe is necessary).
This allows us to appropriately capture the geometry of sub-level sets of $h_{\theta}(y)$ and as a consequence, understand effects of seemingly large perturbations in $y$ which only lead to small changes in the function value.
When tracked in the primal domain, the proof resembles the ``centering'' idea, however, the implicit direction to move $\theta$ along is not easy to come up with (or analyze) just from the primal perspective. 

\parag{Theorem~\ref{thm:bound_informal} -- proof overview. } 
At a high level, in the proof we consider the optimal dual solution $y^\star$ and first identify a vertex $\alpha^\star$ such that $y^\star$ belongs to $C_{\alpha^\star}$ -- the normal cone at $\alpha^\star$ (i.e., the set of directions $z$ such that all $\alpha^\star+tz \in P$ for sufficiently small $t>0$).
Subsequently, a  projection operation of $y^\star$ with respect to the cone $C_{\alpha^\star}$ is used to find a vector $y^{\circ}$ -- a witness for a short and close to optimal solution of the dual problem.
The proof can be decomposed naturally into the following three steps, which we subsequently explain in more detail.
\begin{enumerate}[noitemsep,topsep=0pt,label=(\alph*)]
\item Identify a ``good'' basis for the dual space with respect to $\theta$.
\item Truncate the optimal dual solution with respect to the basis in (a).
\item Establish a bound on the length of the truncated solution.
\end{enumerate}
In what follows, we assume for simplicity that the polytope $\conv(\cF)$ is full-dimensional.
Given any $\theta\in P$ and a point $y^\star \in \R^m$ which satisfies $h_{\theta}(y^\star) \leq g(\theta) + \nfrac{\eps}{2}$ (i.e., is close to optimal\footnote{Note that there might not exist a point $y^\star$ such that $g(\theta)=h_{\theta}(y^\star)$ when $\theta$ is on the boundary of $P$; 
that is why we allow a slack of $\eps/2$.}) we aim to find a point $y^\circ$ whose length is polynomial in $m$ and $\log \nfrac{1}{\eps}$ such that $h_{\theta}(y^\circ) \leq h_{\theta}(y^\star)+\nfrac{\eps}{2}$. 

\parag{Step (a) Good basis.} We first identify a subset $I_0 \subseteq I$ such that $\{a_i : i\in I_0\}$ is a basis of $\R^m$ and $y^\star$ is expressed as a nonnegative linear combination in this basis, i.e., $y^\star = \sum_{i\in I_0} \beta_i a_i$ with $\beta\geq 0$.
The basis $I_0$ is chosen as a basis of tight constraints at a point $\alpha^\star \in \cF$, i.e., $\alpha^\star$ satisfies $\inangle{a_i,\alpha^\star} = b_i$ for $i\in I_0$.
This follows by selecting an $\alpha$ that maximizes $\inangle{\alpha, y^\star}$ over all $\alpha \in \cF$ and invoking Farkas' lemma and Caratheodory's theorem.
This part of the reasoning does not make any assumptions on the polytope and only relies  on the convexity of $P$.

\parag{Step (b) Truncation of coefficients.} Subsequently, we prove that $y^\circ := \sum_{i\in I_0} \min(\Delta, \beta_i) a_i $ satisfies the claim stated above, for a suitable choice of $\Delta$, polynomial in the considered parameters.
The bound $h_{\theta}(y^\circ) \leq h_{\theta}(y^\star)+\nfrac{\eps}{2}$ is proved by replacing $\beta_i$ by $\min(\Delta, \beta_i)$ one by one and showing that the value $h_{\theta}(\cdot)$ does not increase by more than $\nfrac{\eps}{2m}$.
This relies on a careful analysis of the effect such a perturbation has on the function and crucially uses the fact that the coefficients $\beta_i$ for $i\in I_0$ are nonnegative.
Most importantly, we rely on the fact that the coefficients of the inequalities defining $P$ are integral; and hence for any point $\alpha \in \cF$ which does not lie on a facet $\inangle{a_i, x}=b_i$ for some $i\in I_0$ we have $\inangle{\alpha^\star - \alpha, a_i}\geq 1$.

\parag{Step (c) Concluding a polynomial bound.} To bound the length of $y^\circ$ note first that all the vectors $\{a_i\}_{i\in I_0}$ are short, i.e., $\norm{a_i}_2 \leq \norm{a_i}_\infty \cdot m \leq \fc(P) \cdot m = \poly(m)$, because we assume that the unary facet complexity of $P$ is polynomially bounded.
Further, from the triangle inequality, we obtain $\norm{y^{\circ}}\leq m \cdot \Delta \cdot \poly(m) = \poly(m, \log \frac{1}{\eps})$.
Thus, we conclude the proof of Theorem \ref{thm:bound_informal}.

\parag{Theorem~\ref{thm:lower_bound} -- proof overview. } The proof of Theorem~\ref{thm:lower_bound} is based on existence of so-called {\it flat} 0-1 polytopes.
These are polytopes of the form $\conv\{\alpha_0,  \ldots, \alpha_m\}$ with $\alpha_0, \ldots, \alpha_m \in \{0,1\}^m$ such that the distance from $\alpha_0=0$ to the affine subspace $H$ generated by $\alpha_1, \ldots, \alpha_m$ is exponentially small.
The existence of such configurations was proved in~\cite{AV97}.
Given such a polytope we consider the lattice generated on $H$ by the points $\alpha_1, \ldots, \alpha_m$ and construct a new polytope by taking a certain finite subset of such lattice points and the point $0$.
The vertices of the new polytope are still integral and have relatively small entries (polynomial in $m$). 
Moreover, the projection of $0$ onto $H$ lies within the opposite facet.
The family $\cF$ is defined to be all the vertices of the newly constructed polytope, $\theta$ is chosen to be $0$ and we pick $\eps\approx e^{-m}$.

To prove that for every short $y\in \R^m$ the gap between $h_{\theta}(y)-g(\theta)$ is significant we consider the gradient $\nabla_y h_{\theta}(y)$. 
Intuitively, if the gradient is large (in magnitude) at a point $y$, then $y$ cannot be an approximately optimal solution, hence it is enough to show that every short vector $y$ admits a suitably long gradient. 
For this one can show that the gradient at $y$ is  given by $\theta - \theta^y$ where $\theta^y$ is the expectation of a distribution defined by $y$.

In order to make $\theta^y$ $\eps$-close to $\theta=0$, one has to ensure that $\inangle{0,y} - \inangle{\alpha, y} \gtrsim \Omega(1)$ for all $\alpha \in \cF \setminus \{0\}$.
However, by introducing an auxiliary optimization problem (see Fact~\ref{fact:short_y}) we show that this can happen only when $\norm{y}$ is roughly, inverse-proportional to the distance from $0$ to $\conv\inparen{\cF \setminus \{0\}}$.
Since the distance is exponentially small, we arrive at a  lower-bound on $\norm{y}$.

\parag{Theorem~\ref{thm:stability} -- proof overview. } The proof of Theorem~\ref{thm:stability} crucially relies on the bound established in Theorem~\ref{thm:bound_informal}. One starts by picking any ``close-by'' points $\theta_1, \theta_2$ in the polytope $P$, we denote $\norm{\theta_1 - \theta_2}=\eps$. The goal is to show that the distributions $q^{\theta_1}$ and $q^{\theta_2}$ are also close (the distance can be upper-bounded by some function of $\eps$). 

To this end, we first find dual solutions $y_1, y_2\in \R^m$ that have low bit complexity and, roughly, have dual suboptimality gap $\leq \eps$. If we denote by $q^{y_1}$ and $q^{y_2}$ the distributions defined by these two dual solutions, then one can show by a simple calculation that
\begin{equation}\label{eq:bound_kl}
\textstyle KL(q^{\theta_i},q^{y_i}) = h_{\theta_i}(y_i) - g(\theta_i)\leq  \eps \end{equation}
for $i=1,2$, and further, using Pinsker's inequality $\norm{q^{\theta_i} - q^{y_i}}_1 \leq \sqrt{\eps}$. While this gives us some handle on $q^{\theta_1}$ and $q^{\theta_2}$ with respect to the dual solutions $y_1, y_2$, we still do not have any control on the distance between $q^{y_1}$ and $q^{y_2}$. To bypass this problem, we again take advantage of the fact in~\eqref{eq:bound_kl} and consider the gap
$h_{\theta_1}(y_2) - g(\theta_1),$
since this is, roughly, an upper bound on the distance between $q^{\theta_1}$ and $q^{y_2}$, thus showing that it is small would mean we are done. To show that this suboptimality gap is small we use the fact that $h_{\theta_1}(y_1) - g(\theta_1) \leq \eps$ and that $h_{\theta_1}(y_2) \approx h_{\theta_1}(y_1)$, more precisely (by convexity)
$$\textstyle h_{\theta_1}(y_2) - h_{\theta_1}(y_1)\leq \norm{\nabla_y h_{\theta_1}(y_1)}_2 \cdot \norm{y_1 - y_2}_2.$$
Here $\norm{\nabla_y h_{\theta_1}(y_1)}_2$ can be easily bounded by a function of $\eps$ -- because $y\mapsto h_{\theta}(y)$ is a smooth function and $y_1$ is $\eps$-close to the optimal solution. Further, crucially, $\norm{y_1 - y_2}$ is bounded by a polynomial in $m$ and $\log \frac{1}{\eps}$, as a consequence of Theorem~\ref{thm:bound_informal}, and this allows as to arrive at the claimed bound.

\vspace{-4mm}

\section{Proof of the bit complexity upper bound}\label{sec:proof_struct}

We start by reintroducing several notions and restating the structural result, as in Section~\ref{ssec:results} only simplified or informal definitions and theorem statements were given.
Consider a finite subset $\cF \subseteq \Z^m$ of the integer lattice and a positive function $p: \cF \to \R_{>0}$.
We will use $p_\alpha$ to denote the value of the function at a point $\alpha \in \cF$, thus treating $p$ as an $|\cF|-$dimensional vector with coordinates indexed by $\cF$. For any $\theta \in \R^m$ and $y\in \R^m$ we define the following generalized max-entropy program

\begin{equation}
\begin{aligned}
	\max ~~ & \sum_{\alpha \in \cF}q_\alpha \log \frac{p_\alpha}{q_\alpha},&\\
		\st ~~ & \sum_{\alpha \in \cF} q_\alpha \cdot \alpha = \theta,  \; \; \; 
		  \sum_{\alpha \in \cF} q_\alpha=1, \; \; \; 
	 q\geq 0.
\end{aligned}
\label{eq:gen_max_entropy}	
\end{equation}
\noindent
In the case when $p$ is normalized so that $\sum_{\alpha \in \cF} p_\alpha=1$, \eqref{eq:gen_max_entropy} describes the distribution closest to $p$ (in $KL$-distance), whose expectation is equal to $\theta$. 
In the case when $p\equiv 1$ (corresponds to the uniform distribution) the above program asks simply for a max-entropy distribution with expectation $\theta$, as in Section~\ref{ssec:results}.
Let us also extend the definition of the dual program, as introduced in Section~\ref{ssec:results} to capture general functions $p$ not only $p \equiv 1$.

\begin{equation}\label{eq:h}
\begin{aligned}
\textstyle g(\theta) := \inf_{y\in \R^m} h_{\theta}(y):=\inf_{y\in \R^m} \log \inparen{\sum_{\alpha \in \cF} p_\alpha e^{\inangle{\alpha - \theta,y}}}.
\end{aligned}
\end{equation}
\noindent 
The following notion of facet complexity of a polytope plays an important role in our main result.
\begin{definition}[Unary facet complexity]\label{def:fc}
Let $P \subseteq \R^m$ be a convex polytope with integer vertices. Let $M\in \N$ be the smallest integer such that $P$ has a description of the form
\begin{equation}\label{eq:P_desc}
\textstyle P=\{x\in \R^m: \inangle{a_i, x} \leq b_i, \mbox{ for }i\in I\}\cap H
\end{equation}
where $I$ is a finite index set, $a_i \in \Z^m$, $\norm{a_i}_\infty\leq M$ and $b_i\in \R$ for $i\in I$, and $H$ is a linear subspace of $\R^m$. Then we call $M$ the {\em unary facet complexity} of $P$ and denote $\fc(P)=M$.
\end{definition}
 Observe that in a description of a polytope $P$ as in the definition above we could have also included $H$ in the first term of~\eqref{eq:P_desc}, by adding $\inangle{c,x}\leq d$ and $\inangle{c,x} \geq d$, for every equation $\inangle{c,x}=d$ defining $H$. 
 However, the unary facet complexity defined as above might be significantly lower than one measured with respect to all (including equality) constraints and, as it turns out, this is the right measure to study the bit complexity of close to optimal solutions of~\eqref{eq:gen_max_entropy}.

To state our main result also a notion of {\it bit complexity} of a function $p: \cF \to \R_{>0}$ is required. We denote it by $L_p$ and define as
$L_p:=\max_{\alpha \in \cF} |\log p_\alpha|.$
Note that $L_p$ is finite, because we assume that $p_\alpha>0$ for every $\alpha \in \cF$. 
It represents, roughly, the maximum\footnote{Importantly, the complexity measure $L_p$ is the {\it maximum} -- not the {\it total} -- number of bits required to store any of the coefficients. 
The latter is always at least $|\cF|$, hence typically exponential in $m$.} number of bits required to store the binary representation of any $p_\alpha$ (for $\alpha \in \cF$).\footnote{In this paper, we use $\norm{x}$ to denote the Euclidean $\ell_2$ norm of $x$. This choice of a norm is by no means essential.}

\begin{theorem}[Bit complexity upper bound]\label{thm:bound}
Let $\cF$ be any finite subset of $\Z^m$
and let $M$ be the unary facet complexity of $\conv(\cF)$. Then, for every function $p: \cF \to \R_{>0}$ and for every $\eps>0$ there exists a number $R=O(m^{3/2}\cdot M \cdot (L_p+\log \abs{\cF} + \log \frac{m}{\eps}))$ such that 
$$\forall_{\theta\in P}~ \exists_{y\in B(0,R)}~~ h_{\theta}(y) \leq g(\theta) + \eps,$$
where $h$ and $g$ are defined as in~\eqref{eq:h}.
\end{theorem} 
\noindent 
We start by a preliminary lemma which is then used in the proof of Theorem~\ref{thm:bound}.
\begin{lemma}[Good basis]\label{lemma:polytope}
Let $P\subseteq \R^m$ be a polytope $P:=\{x\in \R^m: \inangle{a_i,x} \leq b_i \mbox{ for }i\in I\} \cap H$, where $I$ is a finite index set, $H\subseteq \R^m$ is a linear subspace of $\R^m$ and $a_i \in H$, $b_i \in \R$ for $i\in I$. Let $y\in H$ be any vector, then there exists a vertex $v\in P$ and a subset of the constraints $I_0\subseteq I$ of size $|I_0|\leq \dim(H)$, such that $\inangle{a_i,v}=b_i$ for $i\in I_0$ and there exist non-negative numbers $\{\beta_i\}_{i\in I_0}$ satisfying
$\sum_{i\in I_0} \beta_i a_i = y.$
\end{lemma}
We note that in the case when $H=\R^m$, Lemma~\ref{lemma:polytope} can be also geometrically interpreted as existence of a separating hyperplane for the convex set $P$.
In the proof of Theorem~\ref{thm:bound} we require this slightly more general variant where $H$ is an arbitrary linear subspace.

\begin{proof}
We start by picking $v \in P$ as an $x$ that maximizes $\inangle{x, y}$ over $x\in P$. Note that since $P$ is a polytope (and is compact) such a maximum exists and, we might assume that $v$ is a vertex of $P$.
Given such a $v\in P$ determine all inequalities which are tight at $v$, i.e.,
$I^\star = \{ i\in I: \inangle{a_i, v} = b_i\}.$
Note that $|I^\star|$ might be arbitrarily large, not even polynomially bounded.  We claim that there exist $\{\beta_i\}_{i\in I^\star} $ with $\beta_i\geq 0$ for all $i\in I^\star$, such that 
$y =\sum_{i\in I^\star} \beta_i a_i.$
\noindent 
Suppose it is not the case. Then from the Farkas lemma, there exists a vector $z\in \R^m$ such that:
$$\forall_{i\in I^\star } ~~\inangle{z, a_i} \leq 0 ~~~~~~\mbox{and}~~~~~~~\inangle{z, y} >0.$$
Note also that we may assume that $z\in H$, by projecting $z$ orthogonally onto $H$ if necessary. Further, the above is true also for $\delta \cdot z$ (in place of $z$) for an arbitrarily small $\delta>0$. In other words, we can take $z$ of arbitrarily small norm. Hence we obtain that the cone
$C=\{u\in H: ~\forall_{i\in I^\star} ~ \inangle{u, a_i} \leq 0\}$
contains arbitrarily short vectors $z$ with $\inangle{z,y}>0$. Note that since $I^\star$ is the collection of all inequalities tight at $v$, it follows that every point in $H$, which is sufficiently close to $v$ and satisfies the inequalities in $I^\star$ belongs itself to $P$. In other words, there exists a $\delta>0$ (might be exponentially small) such that
$(v+C) \cap B(v, \delta) \subseteq P,$
where $v+C = \{v+u:u\in C\}$ and $B(v, \delta)=\{x\in \R^m: \norm{x-v}\leq \delta\}$. Combining this with our previous observation regarding the cone $C$ it follows that there exists $z\in H$ such that $\mu:=v+z \in P$ and $\inangle{z,y}>0$ and hence $\inangle{\mu, y}>\inangle{v, y}. $
This contradicts our choice of $v\in P$;
$\inangle{\alpha, y} >  \inangle{{v}, y}.$
Knowing that $y$ belongs to the cone generated by $\{a_i\}_{i\in I^\star}$ we can apply Caratheodory's theorem to reduce the number of nonzero coefficients in the resulting conic combination. Indeed there exists a set $I_0\subseteq I^\star$, such that $|I_0|\leq \mathrm{dim}(H)$ and non-negative $\{\beta_i\}_{i\in I_0}$ (possibly different $\beta_i$s than obtained above) such that
$y=\sum_{i\in I_0} \beta_i a_i'.$
\end{proof}

\begin{proofof}{of Theorem \ref{thm:bound}}
Before we proceed with the argument let us first observe that one can assume without loss of generality that $0 \in \cF$. This follows from the ``shift invariance'' of our problem. Indeed if we consider $\cF$ and $\cF'=\cF+\gamma = \{\alpha + \gamma: \alpha \in \cF\}$, then the corresponding functions $h$ ad $h'$ satisfy
$h_{\theta}(y) = h'_{\theta+\gamma} (y)$
for every $y\in \R^m$. Hence, by shifting $\cF$ by $\gamma=-\alpha$ for some $\alpha \in \cF$ we obtain an equivalent instance of our problem with $0 \in \cF$. It follows in particular, that the affine subspace $H$ on which $P$ is full-dimensional is now a {\it linear subspace} of $\R^m$.

Fix $\theta\in P$ and let $y^\star$ be such that
$ h_{\theta}(y^\star)\leq g(\theta)+\frac{\eps}{2}.$
Note that we may assume that $y^\star \in H$, by projecting it orthogonally onto $H$ (which does not alter the value). Further, note that by denoting by $a_i'\in H$ (for $i\in I$) the orthogonal projection of $a_i$ onto $H$, the polytope  $P$ can be equivalently written  as
$P=\{x\in H: \inangle{a_i', x}\leq b_i\}.$
By applying Lemma~\ref{lemma:polytope} for $y^\star$ and $P$ we obtain a vertex $\alpha^\star \in \cF$ and a subset of constraints $I_0$ of size at most $\dim(H) \leq m$, tight at $\alpha^\star$ such that
$y^\star =\sum_{i\in I_0} \beta_i a_i',$
for some non-negative scalars $\{\beta_i\}_{i\in I_0}$. We prove that by modifying the coefficients in the above conic combination we can obtain a point 
$y'=\sum_{i\in I_0} \beta_i' a_i'$
(with $\beta_i' \geq 0$ for $i\in I_0$) such that the norm of $y'$ is small (polynomial in $L_p,m,\log d, M$ and $\log \frac{1}{\eps}$) and 
$$h_{\theta}(y') \leq h_{\theta}(y^\star) +\frac{\eps}{2}\leq g(\theta)+\eps.$$
Showing that will complete the argument. 
Let $\Delta >0$ be a certain number (to be specified later), polynomial in $L_p$ and $\log \frac{1}{\eps}$. We define $\beta_i':=\min(\Delta, \beta_i)$ and prove that the point $y'=\sum_{i\in I_0} \beta_i' a_i'$
satisfies the above claim.

To this end, we prove that by changing one coordinate $i_0$, from $\beta_{i_0}>\Delta$ to $\Delta$ we cause only a slight increase in the value of $h_{\theta}(y)$. In other words, by taking $y^\star$ as before and $y'=y^\star - (\beta_{i_0}-\Delta) a_i$ 
we want to show that
$$\textstyle \log\inparen{\sum_{\alpha\in \cF} p_\alpha e^{\inangle{\alpha, y'} - \inangle{\theta, y'}}}\leq\log\inparen{\sum_{\alpha\in \cF} p_\alpha e^{\inangle{\alpha, y^\star} - \inangle{\theta, y^\star}}}+ \frac{\eps}{2m}$$
Towards this, define
$\cF_0 = \{\alpha\in \cF: \inangle{\alpha, a_{i_0}}=b_{i_0}\}.$
Below we analyze separately the effect of changing $y$ to $y'$ on the terms $p_\alpha e^{\inangle{\alpha, y}}$ for $\alpha \in \cF_0$ and for $\alpha \in \cF \setminus \cF_0$.  

\parag{Case 1: $\cF \setminus \cF_0$.} 
Consider any $\alpha\in \cF \setminus \cF_0$. We have
\begin{align*}
\inangle{\alpha, y'} - \inangle{\alpha^\star, y'} &= \sum_{i\in I_0} \beta_i' \inangle{\alpha-\alpha^\star, a_i'}\\
 &= \sum_{i\in I_0} \beta_i' \inangle{\alpha-\alpha^\star, a_i}\\
&\leq \Delta \inangle{\alpha - \alpha^\star, a_{i_0}}\\
&\leq -\Delta.
\end{align*}
In the above, we used the fact that $a_i'$ is the projection of $a_i$ onto $H$, $\inangle{\alpha - \alpha^\star, a_i}\leq 0$ for every $i\in I_0$ and that $\inangle{\alpha - \alpha^\star, a_{i_0}}$ is a negative integer. This implies in particular, that
$$\frac{p_\alpha e^{\inangle{\alpha, y'} - \inangle{\theta, y'}}}{p_{\alpha^\star} e^{\inangle{\alpha^\star, y'} - \inangle{\theta, y'}}}\leq \frac{p_{\alpha}}{p_{\alpha^\star}}e^{-\Delta}$$
and hence:
\begin{align*}
\sum_{\alpha\in \cF} p_\alpha e^{\inangle{\alpha, y'} - \inangle{\theta, y'}}&=\sum_{\alpha\in \cF_0} p_\alpha e^{\inangle{\alpha, y'} - \inangle{\theta, y'}}+\sum_{\alpha\in \cF\setminus \cF_0} p_\alpha e^{\inangle{\alpha, y'} - \inangle{\theta, y'}}\\
&\leq \inparen{\sum_{\alpha\in \cF_0} p_\alpha e^{\inangle{\alpha, y'} - \inangle{\theta, y'}}}\inparen{1+|\cF \setminus \cF_0|\frac{\max_{\alpha \in \cF \setminus \cF_0}p_{\alpha}}{p_{\alpha^\star}}e^{-\Delta}}
\end{align*}
Note that for any $\alpha\in \cF$ we have $e^{-L_p} \leq p_\alpha \leq e^{L_p}$, hence we can pick $\Delta:=\log \abs{\cF} + 2 L_p + \log m + \log \frac{1}{\eps}$ to guarantee
$$|\cF \setminus \cF_0|\frac{\max_{\alpha \in \cF \setminus \cF_0}p_{\alpha}}{p_{\alpha^\star}}e^{-\Delta}\leq \frac{\eps}{2m}.$$
For such a choice of $\Delta$ we have:
$$h_{\theta}(y') \leq \log  \inparen{\sum_{\alpha\in \cF_0} p_\alpha e^{\inangle{\alpha, y'} - \inangle{\theta, y'}}} + \log\inparen{1+\frac{\eps}{2m}} \leq \log  \inparen{\sum_{\alpha\in \cF_0} p_\alpha e^{\inangle{\alpha, y'} - \inangle{\theta, y'}}}+\frac{\eps}{2m}.$$

\parag{Case 2: $\cF_0$.} 
Consider now $\alpha \in \cF_0$, we have
$$\inangle{\alpha, y'} - \inangle{\theta, y'} = \inangle{\alpha, y^\star} - \inangle{\theta, y^\star} - (\beta_i-\Delta) \inangle{\alpha-\theta, a_i}\leq \inangle{\alpha, y^\star} - \inangle{\theta, y^\star}$$
as $\inangle{\theta, a_i}\leq \inangle{\alpha, a_i}=b_i$ (because $\theta\in P$). Consequently, we obtain

\begin{align*}
\textstyle h_{\theta}(y')& \leq \log  \inparen{\sum_{\alpha\in \cF_0} p_\alpha e^{\inangle{\alpha, y'} - \inangle{\theta, y'}}}+\frac{\eps}{2m}\\
 &\leq  \inparen{\sum_{\alpha\in \cF_0} p_\alpha e^{\inangle{\alpha, y^\star} - \textstyle \inangle{\theta, y^\star}}}+\frac{\eps}{2m}\\
 & \leq h_{\theta}(y^\star)+\frac{\eps}{2m}.
 \end{align*}
\noindent 
It remains to argue that after performing the above procedure, the norm of $y'$ is small. We have
$$\textstyle \norm{y'}=\norm{\sum_{i\in I_0}^m \beta_i' a_i'} \leq \sum_{i\in I_0} \beta_i' \norm{a_i'} \leq m \cdot \Delta \cdot (\sqrt{m}\cdot M)= m^{3/2} M \Delta.$$
In the above, we used the fact that since $a_i'$ is a projection of $a_i$ onto H (for any $i\in I_0$) we have $\norm{a_i'} \leq \norm{a_i} \leq \sqrt{m}\cdot M.$
\end{proofof}

\noindent
Below we present a useful corollary of Theorem~\ref{thm:bound} which is often convenient in applications. To state it, denote by $d\in \R_{\geq 0}$ the diameter (in the Euclidean norm) of the set $\cF$.
\begin{corollary}\label{cor:rep}
Under the assumptions of Theorem~\ref{thm:bound},  for every function $p: \cF \to \R_{>0}$, for every $\eps>0$ there exists a number $R>0$ which is polynomial in $m, \log d, M, L_p$ and $\log \frac{1}{\eps}$ such that 
$$\textstyle \forall_{\theta\in P}~ \exists_{y\in B(0,R)}~~ \norm{\frac{\sum_{\alpha \in \cF} p_\alpha e^{\inangle{\alpha,y}} \cdot \alpha}{\sum_{\alpha \in \cF} p_\alpha e^{\inangle{\alpha,y}}} - \theta} <\eps.$$
\end{corollary}

\begin{proofof}{of Corollary~\ref{cor:rep}}
It is enough to establish it for $\theta \in \interior(P)$ (relative interior is meant here, i.e., interior of $P$ when restricted to $H$). To this end, consider $y^\star$ to be 
$y^\star = \argmin_{y\in \R^m} h_{\theta}(y) = g(\theta).$
It is not hard to prove that such a $y^\star$ exists, i.e., the minimum is attained, for $\theta$ in the relative interior of $P$. Consider now the gradient of $h$ with $\theta$ fixed:
\begin{align}\label{eq:grad}
\textstyle \nabla_{y} h_{\theta}(y) = \frac{\sum_{\alpha \in \cF} p_\alpha e^{\inangle{\alpha,y}} \cdot \alpha}{\sum_{\alpha \in \cF} p_\alpha e^{\inangle{\alpha,y}}} - \theta.
\end{align}
To conclude the corollary from Theorem~\ref{thm:bound} one has to prove that if the value at a point is close to optimal, then the gradient is short. Towards this, we show that $y \mapsto h_{\theta}(y)$ is $L-$smooth (for some polynomially bounded $L$), i.e.,
$\norm{\nabla_y h_{\theta}(y_1) - \nabla_y h_{\theta}(y_2)} \leq L \norm{y_1 - y_2}.$
This can be deduced from the fact that the Hessian matrix of $h$, $\nabla^2_y h_{\theta}(y)$ has polynomially bounded entries (and the bound depends neither on $y$ nor on $\theta$). Indeed, under the notation that $\alpha' = \alpha - \theta$ for $\alpha \in \cF$ we obtain
$$ \textstyle \inparen{\nabla^2_y h_{\theta}(y)}_{i,j} = \frac{\sum_{\alpha \in \cF} p_\alpha e^{\inangle{\alpha,y}}\alpha'_i \alpha'_j }{\sum_{\alpha \in \cF} p_\alpha e^{\inangle{\alpha,y}}} - \frac{\inparen{\sum_{\alpha \in \cF} p_\alpha e^{\inangle{\alpha,y}}\alpha'_i} \inparen{\sum_{\alpha \in \cF} p_\alpha e^{\inangle{\alpha,y}}\alpha'_j}}{\inparen{\sum_{\alpha \in \cF} p_\alpha e^{\inangle{\alpha,y}}}^2},$$
hence it follows
$$\textstyle \abs{\inparen{\nabla^2_y h_{\theta}(y)}_{i,j}} \leq 2\max_{\alpha \in \cF} \norm{\alpha - \theta}^2\leq 2 d^2.$$
Hence the function $y \mapsto h_{\theta}(y)$ is $L:=2d^2-$smooth.  Now, it is well known that for a convex, $L-$smooth function, we have (see e.g.~\cite{Nesterov14}):
$$h\inparen{\theta, y+\frac{1}{2L} v}\leq h_{\theta}(y)-\frac{1}{4L}\norm{v}^2,$$
where $v=\nabla_y h_{\theta}(y)$. Hence, if $y$ is as in Theorem~\ref{thm:bound} then we obtain
$\norm{\nabla_y h_{\theta}(y)}^2\leq 4L \eps$
and consequently
$$\norm{\nabla_y h_{\theta}(y)}\leq 4d \eps^{1/2}.$$
The corollary follows by combining the above obtained bound with~\eqref{eq:grad}.
\end{proofof}

\section{Stability of max-entropy distributions}\label{sec:stability}
The purpose of this section is to first restate Theorem~\ref{thm:stability} in a more general (weighted) form and then provide a proof of this generalized theorem.
We let $\cF \subseteq \Z^m$ to be any finite family of integer vectors and fix any function $p:\cF \to \R_{>0}$. As usually, we denote the convex hull of $\cF$ by $P$.
For brevity, we introduce the following handy notation:
\begin{enumerate}
\item Given $\theta\in P$ we denote by $q^\theta$ the optimal solution of the max-entropy convex program~\eqref{eq:gen_max_entropy}.
\item Given $y\in \R^m$ let $q^y$ be the distribution over $\cF$ defined as $$q_\alpha^y \propto p_{\alpha} e^{\inangle{y,\alpha}}~~~~~\mbox{ for }\alpha \in \cF.$$
Moreover, we let $\theta^y:=\sum_{\alpha \in \cF} q_\alpha^y \alpha.$
\end{enumerate}
Note in particular that if the dual program~\eqref{eq:h} has an optimal solution $y$ then the KKT conditions imply that $q^y = q^\theta$ and it follows from feasibility that $\theta^y = \theta$.
We are thus ready to state the generalized version of Theorem~\ref{thm:stability}.

\begin{theorem}[\bf Polynomial stability of max-entropy distributions]\label{thm:stability_gen}
Let $\cF$ be any finite subset of $\Z^m$
 and let $d\in \R_{\geq 0}$ be its diameter (in the $\ell_2$-norm)
and let $M$ be the unary facet complexity of $\conv(\cF)$. Then, for every function $p: \cF \to \R_{>0}$ and for every $\eps>0$ there exists a number $R=O(m^{3/2}\cdot M \cdot (L_p+\log \abs{\cF} + \log \frac{dm}{\eps}))$ such that  for every $\theta_1, \theta_2 \in P$, if $\norm{\theta_1 - \theta_2}_1 \leq \eps$ then 
$$\norm{q^{\theta_1} - q^{\theta_2}}_1\leq \sqrt{R} \sqrt{\eps}.$$
\end{theorem}

\noindent
Before we proceed to the proof we need to show a few simple properties of the max-entropy convex program and especially how do the primal and dual variables relate with each other. 
We begin by an elementary, yet important technical lemma. Below, whenever we refer to $h_{\theta}(y)$ or $g(\theta)$ we mean the dual objective, and its minimum respectively, as defined in~\eqref{eq:h}, where $\cF$ and $p$ will be clear from the context.

\begin{lemma}[KL-distance between feasible distributions of the max-entropy program]\label{lemma:equality}
Consider any finite family $\cF \subseteq \Z^m$ along with a positive function $p:\cF \to \R_{>0}$ then, for any $\theta \in P$, any $y\in \R^m$ and any distribution $q$ that is feasible for the max-entropy program~\eqref{eq:gen_max_entropy}, we have
$$KL(q,q^y) = h_{\theta}(y)-\sum_{\alpha \in \cF} q_\alpha \log \frac{p_\alpha}{q_\alpha}.$$
\end{lemma}
\begin{proof}
The proof follows by a direct calculation, it is enough to use the fact that $\sum_{\alpha \in \cF} q_\alpha \alpha = \theta.$
\end{proof}

\noindent
The next step is to show some basic properties of close-to-optimal solutions to the dual program, and more precisely how does a small suboptimality gap in the dual sense translate to the primal world. 

\begin{lemma}[Basic properties of close-to-optimal solutions to the dual program]\label{fact:basic}
Let $\cF$ be any finite subset of $\Z^m$ and let $d\in \R_{\geq 0}$ be its diameter, let $M$ be the unary facet complexity of $\conv(\cF)$. Then, for every function $p: \cF \to \R_{>0}$, for every $\theta\in P$ and for every $\eps>0$ 
if for $y\in \R^m$ we have $h_{\theta}(y)\leq  g(\theta)+\eps$ then 
\begin{enumerate}
\item $\norm{q^\theta -q^y}_1\leq\sqrt{2 \eps}.$
\item $\norm{\theta^y - \theta}_2 \leq 2d \sqrt{\eps}.$
\end{enumerate}

\end{lemma}
\begin{proof}
\paragraph{Part 1.} We start by applying Lemma~\ref{lemma:equality} for $q:=q^\theta$ to obtain
$$h_{\theta}(y) - g(\theta) = KL(q^\theta, q^y)$$
and thus by Pinsker inequality
$$\norm{q^\theta - q^y}_1 \leq \sqrt{2 KL(q^\theta, q^y)}.$$
\paragraph{Part 2.}
One can start by observing that 
$$\theta^y - \theta = \nabla_y h_{\theta}(y).$$
Further, since the function $y\mapsto h_{\theta}(y)$ is $2d^2$-smooth, it is well known that one step of the gradient descent algorithm started at $y$ (with step size $\frac{1}{4d^2}$) produces a point $y'$ with value
$$h_{\theta}(y') \leq h_{\theta}(y) - \frac{\norm{\nabla_y h_{\theta}(y)}_2^2}{4d^2},$$
and consequently $\norm{\nabla_y h_{\theta}(y)}_2\leq  2d \sqrt{\eps}$.
\end{proof}
\noindent
We are now ready to prove the theorem.

\begin{proofof}{ of Theorem~\ref{thm:stability_gen}}
Using Theorem~\ref{thm:bound} we can find an $R=O(m^{3/2}\cdot M \cdot (L_p+\log \abs{\cF} + \log \frac{m}{\delta}))$ and $y_1, y_2 \in B(0,R)$ such that
$$h_{\theta_1}(y_1) - g(\theta_1) \leq \delta ~~~~~ \mbox{and}~~~~~~h_{\theta_2}(y_2) - g(\theta_2)\leq \delta.$$
The value of $\delta>0$ will be chosen later (as a function of $\eps>0$). We prove that $h_{\theta_1}(y_2) - g(\theta_1)$ is small and thus by Fact~\ref{fact:basic} we deduce that $q^{y_2}$ and $q^{\theta_1}$ are close, thus also $q^{\theta_1}$ and $q^{\theta_2}$ are close. The details follow.

Since the function $y\mapsto h_{\theta_1}(y)$ is convex, the first order convexity condition yields

\begin{equation}
h_{\theta_1}(y_1) \geq h_{\theta_1}(y_2) + \inangle{\nabla_y h_{\theta_1} (y_2), y_1-y_2}
\end{equation}
and thus
\begin{align*}
h_{\theta_1}(y_2) &\leq h_{\theta_1} (y_1)+ \inangle{\nabla_y h_{\theta_1}(y_2), y_2-y_1}\\
&\leq h_{\theta_1}(y_1)+\norm{\nabla_yh (\theta_1, y_2)}_2\cdot \norm{y_2 - y_1}_2\\
& \leq g(\theta_1)+\delta+2R\cdot \norm{\nabla_y h_{\theta_1}(y_2)}_2
\end{align*}
By a simple calculation:
\begin{align*}
 \nabla_y h_{\theta_1}(y_2)&=\theta^{y_2} - \theta_1\\
 &=(\theta^{y_2} - \theta_2)+(\theta_2 - \theta_1)
\end{align*}
and further, by taking the norm and using Lemma~\ref{fact:basic}

$$ \norm{\nabla_yh (\theta_1, y_2)}_2 \leq \norm{\theta^{y_2} - \theta_2}_2 + \norm{\theta_2 - \theta_1}_2 \leq 4d\sqrt{\delta}+\eps.$$
\noindent
Thus consequently we obtain:
$$h_{\theta_1}(y_2) \leq g(\theta_1) + 10Rd\sqrt{\eps}.$$
Now, by Lemma~\ref{fact:basic} it follows that
$$\norm{q^{\theta_1} - q^{y_2}}_1 \leq O(\sqrt{R}) \sqrt{4d\sqrt{\delta} + \eps}$$
And further, again by Lemma~\ref{fact:basic} we obtain
$$\norm{q^{\theta_2} - q^{y_2}}_1 \leq \sqrt{2\delta}$$
\noindent
and hence by the triangle inequality
$$\norm{q^{\theta_1}-q^{\theta_2}}_1 \leq O(\sqrt{R}) \sqrt{4d\sqrt{\delta} + \eps}+  \sqrt{2\delta}$$
Now, by taking $\delta>0$ such that $4d\sqrt{\delta} = \eps$ we obtain that 
$$\norm{q^{\theta_1}-q^{\theta_2}}_1 \leq O(\sqrt{R\eps}).$$
\end{proofof}

\section{Proof of the bit complexity lower bound}\label{sec:proof_lower}
We start by proving a technical fact which will be useful in establishing the large bit complexity example.

\begin{fact}\label{fact:short_y}
Let $v_1, v_2, \ldots, v_N \subseteq \R^m$ be a set of vectors and denote $\delta:=\dist(0,\conv(v_1, v_2, \ldots, v_N))$. Assume that $\delta>0$ and consider the optimal value of the following optimization problem:
$$\tau = \min \inbraces{\norm{y}:y\in \R^m,  \inangle{y,v_i}\leq -1 \mbox{ for all }i=1,2, \ldots, N},$$
then $\tau\geq \frac{1}{\delta}$.
\end{fact}
\begin{proof}
We formulate the problem as a convex, quadratic program with linear constraints:
\begin{equation}
\begin{aligned}
	\min ~~ &\norm{y}^2,&\\
		\st ~ & \inangle{y,\alpha}\leq -1 ~~~\mbox{for all }\alpha \in \cF'.
\end{aligned}
\label{eq:short_y}	
\end{equation}
To derive a lower bound on the optimal value of~\eqref{eq:short_y} we consider the dual program:
\begin{equation}
\begin{aligned}
	\max ~~ &\sum_{i=1}^N \lambda_i - \frac{1}{4} \norm{\sum_{i=1}^N \lambda_i v_i}^2,&\\
		\st ~ & \lambda_i\geq 0 ~~~\mbox{for all }i=1,2, \ldots, N.
\end{aligned}
\label{eq:dual_short_y}	
\end{equation}
From weak duality we know that the optimal value of~\eqref{eq:dual_short_y} is a lower bound to~\eqref{eq:short_y} thus we just need to provide a feasible solution to the dual program. 
To this end, let $v$ be the shortest vector in the convex hull of $v_1, v_2, \ldots, v_N$, i.e., $v\in  \conv(v_1, v_2, \ldots, v_N)$  and $\norm{v}=\delta$. $v$ can be written as
$$v= \sum_{i=1}^N \mu_i v_i$$
for some $\mu\geq 0$ with $\sum_{i=1}^N \mu_i=1$. Consider now $\lambda:= \frac{2}{\delta^2} \mu \in \R_{\geq 0}^N$.  The dual objective value for $\lambda$ is
$$\sum_{i=1}^N \lambda_i - \frac{1}{4} \norm{\sum_{i=1}^N \lambda_i v_i}^2 = \frac{2}{\delta^2}  - \frac{1}{4}\cdot  \frac{4}{\delta^4} \norm{\sum_{i=1}^N \mu_i v_i}^2 = \frac{1}{\delta^2}.$$
This provides us with a lower bound of $\frac{1}{\delta^2}$ on the optimal value of~\eqref{eq:short_y} and thus a lower bound of $\frac{1}{\delta}$ on the optimization problem as in the statement of the Fact.
\end{proof}

\begin{remark}
It is not hard to prove that in Fact~\ref{fact:short_y} the value $\frac{1}{\tau}$ is not only a lower bound but is actually equal to the optimal value of the considered optimization problem. This can be established by plugging in an appropriate scaling of the shortest vector in $\conv(v_1, v_2, \ldots, v_N)$ for $y$.
\end{remark}

\begin{proofof}{of Theorem~\ref{thm:lower_bound}}
Our construction of $\cF$ is based on existence of ``flat'' 0-1 polytopes, as established by~\cite{AV97}. There exist $m+1$ affinely independent points $\alpha_0, \alpha_1, \alpha_2, \ldots, \alpha_m \in \{0,1\}^m$ such that if we let $H=\alpha_1+\lspan \{ \alpha_2-\alpha_1, \ldots, \alpha_m-\alpha_1\}$ (the $(m-1)-$dimensional affine subspace containing all points $\alpha_1, \ldots, \alpha_m$) then
$$\dist(\alpha_0,H)=e^{- \Omega(m \log m)}.$$
Without loss of generality we assume that $\alpha_0 = 0$.  Let $y\in H$ be the projection of $\alpha_0=0$ onto $H$, i.e., a point such that
$$\inangle{y, \alpha_i -y } = 0 ~~~~~~~~~\mbox{for every }i=1,2, \ldots, m.$$
Consider the lattice 
$$L=\alpha_1 + \sum_{i=2}^m (\alpha_i - \alpha_1)\cdot \Z,$$
with the origin at $\alpha_1$, with basis $\{(\alpha_i - \alpha_1)\}_{2\leq i \leq m}$. Since the hyperplane $H$ is covered by disjoint copies of the fundamental parallelepiped 
$$F:=\inbraces{ \sum_{i=2}^m \beta_i (\alpha_i - \alpha_1): \beta_2, \ldots, \beta_m \in [0,1)},$$
 there is an integer translation of $F$ which contains the point $y$. More formally, there exists an integer vector $\gamma \in \Z^m$ such that
$$y \in \gamma + F \subseteq H.$$
Note now that by denoting $F'=\gamma+F$ we obtain
$$\diam(F') = \diam(F) \leq m^{3/2},$$
since every vertex of $F$ has integer coordinates in the range $[0,m]$. Let $\alpha \in \Z^m$ be now any of the $2^m$ vertices of $F'$. Since $y$ belongs to $F'$ we have 
$$\norm{\alpha} \leq \norm{y} + \diam(F') \leq O(1) +m^{3/2}.$$
Let $\cF'$ be a subset of $\Z^m$ consisting of all the vertices of $F'$ and let $\cF := \cF' \cup \{0\}$. Further, define $\theta := 0$. We prove that the conclusion of the Lemma holds under such a choice of $\cF$ and $\theta$. 

Towards this, we first note that affine hull of $\cF'$ is equal to $H$, as $F' \subseteq H$ and its vertices clearly span $H$. Moreover the point $y$, which is the projection of $0$ onto $H$ belongs to the convex hull of $\cF'$. Let $\delta>0$ be the distance between $0$ and $H$ and let $a\in \R^m$ (with $\norm{a}=1$) be the normal vector of $H$, in other words
$$H=\{x\in \R^m: a^\top x = \delta\}.$$
Note that the gradient of $h_{\theta}(y)$ with respect to $y$ is given by:
$$\nabla_y h_{\theta}(y) = \frac{\sum_{\alpha \in \cF} \alpha \cdot  e^{\inangle{y,\alpha}}}{\sum_{\alpha \in \cF} e^{\inangle{y,\alpha}}}.$$
Since $h$ is $L-$smooth for some $L=\poly(m)$ (see the proof of Corollary~\ref{cor:rep}), we know that points with a large-magnitude gradient cannot be close to optimal. Quantitatively, we have
$$\abs{g(\theta) - h_{\theta}(y)} \geq \frac{\norm{\nabla_y h_{\theta}(y)}^2}{L}.$$
Thus to prove that $\abs{g(\theta) - h_{\theta}(y)} \geq \eps$ (for some $\eps>0$) it is enough to prove that $\norm{\nabla_y h_{\theta}(y)}\geq \sqrt{\eps L}$. 
We pick $\eps$ to be $\frac{\delta^2}{e^4 \cdot L}$, then $\eps = e^{-O(m \log m)}$.   Moreover, the condition $\abs{g(\theta) - h_{\theta}(y)}<\eps$ implies that $\norm{\nabla_y h_{\theta}(y)}< \frac{\delta}{e^2}$. We prove that the latter is possible only when $\norm{y} \geq e^{\Omega(m \log m)}$.

Indeed, assume that $\norm{\nabla_y h_{\theta}(y)}< \frac{\delta}{e^2}$. By the Cauchy-Schwarz inequality we have
$$\inangle{a, \nabla_y h_{\theta}(y)}\leq \norm{\nabla_y h_{\theta}(y)} \cdot \norm{a} = \norm{\nabla_y h_{\theta}(y)} $$
and moreover

$$\inangle{a, \nabla_y h_{\theta}(y)} = \frac{0+\sum_{\alpha \in \cF'} \inangle{a, \alpha}  e^{\inangle{y,\alpha}}}{\sum_{\alpha \in \cF} e^{\inangle{y,\alpha}}}=\delta\cdot  \frac{\sum_{\alpha \in \cF'} e^{\inangle{y,\alpha}}}{1+\sum_{\alpha \in \cF'} e^{\inangle{y,\alpha}}}.$$
It follows that
$$\frac{\sum_{\alpha \in \cF'} e^{\inangle{y,\alpha}}}{1+\sum_{\alpha \in \cF'} e^{\inangle{y,\alpha}}}< \frac{1}{e^2},$$
and consequently $\sum_{\alpha \in \cF'} e^{\inangle{y,\alpha}}< \frac{1}{e}$, which implies in particular, that
\begin{equation}\label{eq:nec_y}
\forall_{\alpha \in \cF'} ~~~~ \inangle{y,\alpha} \leq -1.
\end{equation}
The question we would like to answer is: what is the shortest $y\in \R^m$ which satisfies condition~\eqref{eq:nec_y}? This will give us a lower bound on $\norm{y}$ satisfying $\abs{g(\theta) - h_{\theta}(y)}<\eps$. To answer this question, we apply Fact~\ref{fact:short_y} and conclude that every such $y$ has length at least $\frac{1}{\delta}$. As $\delta = e^{-\Omega(m \log m)}$ we conclude that the optimal solution $y^\star$ to~\eqref{eq:short_y} satisfies $\norm{y^\star}=e^{\Omega(m \log m)}$ and the Lemma follows by contraposition.
\end{proofof}

\section{Computability of max-entropy distributions}\label{sec:proof_comp}
This section is devoted to proving computability of maximum entropy distributions, assuming that an appropriate counting oracle is provided. 
The main component in the argument is the bound on the magnitude of almost optimal dual solutions established in Theorem~\ref{thm:bound}.

\begin{theorem}\label{thm:comp_entro}
Let $\cF$ be any finite subset of $\Z^m$ and let $d\in \R_{\geq 0}$ be its diameter, let $M$ be the unary facet complexity of $\conv(\cF)$.  Then, there exists an algorithm such that given a probability distribution $p$ on $\cF$ (via an evaluation oracle for $g_p$), $\theta\in P$ and an $\eps>0$, computes  a vector $y \in \R^m$ with $\norm{y} \leq \poly\inparen{m,M, \log d,L_p,\log \frac{1}{\eps}}$ such that
$$\norm{q^y - q^\star}_1 <\eps,$$
where $q^\star$ is the optimal solution to~\eqref{eq:gen_max_entropy} and $q^y$ is a distribution over $\cF$ defined as $q^y_\alpha = \frac{p_\alpha e^{\inangle{\alpha, y}}}{\sum_{\cB \in \cF} p_\beta e^{\inangle{\beta, y}}},$ for $\alpha \in \cF$. The running time of the algorithm is polynomial in $m,M, d, L_p, \log \frac{1}{\eps}$. 
\end{theorem}
\noindent 
In the above $g_p$ is a generalized counting function.
An oracle for $g_p$ is then simply defined as a procedure that given an $x>0$ outputs
$$g_p(x):=\sum_{\alpha \in \cF} p_\alpha x^\alpha,~~~~~~~\mbox{where}~~~x^\alpha:=\prod_{i=1}^m x_i^{\alpha_i}.$$
It was shown in~\cite{SinghV14} that efficient (polynomial time) computability of max-entropy distribution essentially implies the existence of polynomial time oracles for $g_p$. 
Hence, in principle, the existence of such oracles is assumed without loss of generality.

\begin{proof}
For convenience without loss of generality we might assume that $\cF \subseteq \N^m$ and in fact even $\cF \subseteq [0,2d]^m \cap \N^m$ where $d$ is the diameter of $\cF$.
This is because one can always shift the set $\cF$ (together with $\theta$) and this operation does not affect the problem nor its parameters.
To obtain a vector $y$, as required, we solve the dual program up to the desired precision (below we explain why is this enough). Recall that the dual program to~\eqref{eq:gen_max_entropy} is given by
\begin{equation}
\begin{aligned}
	\inf ~~& h(\theta, y)=\log\inparen{\sum_{\alpha \in \cF} p_\alpha e^{\inangle{\alpha-\theta, y}}},&\\
		\st ~~~ &y\in \R^m.
\end{aligned}
\label{eq:entropy_dual}	
\end{equation}
By a direct calculation one can show that for every feasible solution $q$ of the primal problem~\eqref{eq:gen_max_entropy}
$$KL(q,q^y) = h(\theta,y)-\sum_{\alpha \in \cF} q_\alpha \log \frac{p_\alpha}{q_\alpha},$$
where $KL(\cdot, \cdot)$ denotes the $KL$-divergence. In particular for $q:=q^\star$
$$h(\theta, y) - g(\theta) = KL(q^\star, q^y).$$
This means that in order to obtain a distribution $q^y$ being $\eps-$close in the KL-distance to the max-entropy distribution $q^\star$ it is enough to find an $\eps-$optimal solution to the dual program. 
Moreover, from Pinsker's inequality we have
$$\norm{q-q^\star}_1 \leq \sqrt{2\cdot KL(q^\star, q)}$$
hence it suffices to find a solution $y$ to the dual program which is $\delta:=\Theta(\eps^2)$-optimal, to guarantee that $\norm{q^\star - q^y}_1 < \eps$.

To find a $\delta$-optimal solution to the dual problem we apply the ellipsoid method. 
First of all we note that $h(\theta, y)$ is a convex function of $y$ (this follows from H\"older's inequality), which is the first requirement for the ellipsoid method to be applicable.
It follows from Theorem~\ref{thm:comp_entro} that in order to find a $\delta$-optimal solution to $\inf_{y\in \R^m} h(\theta,y)$ it is sufficient to solve
$$\min_{y\in B(0,R)} h(\theta,y),$$
where $R$ is a certain bound, polynomial in $m, M, \log d, L_p$ and $\log \frac{1}{\delta}$.
Now, following the treatment of the ellipsoid method in~\cite{BentalN12} (Theorem 8.2.1) it remains to address the following issues.
\begin{enumerate}
\item Construct a first order oracle for $y \mapsto h(\theta,y)$, i.e., an efficient way to evaluate values $h(\theta,y)$ and gradients $\nabla_y h(\theta,y)$ of this function. 
\item Provide an enclosing ball -- containing the domain and a ball contained in the domain.
\item Provide a bound on the gap between the maximum and minimum value of $h(\theta, \cdot)$ in $B(0,R)$. More precisely, defining $D:=\max_{y\in B(0,R)} h(\theta, y) - \min_{y\in B(0,R)} h(\theta,y)$ we would like $\log D$ to be polynomially bounded.
\item Provide a separation oracle for the domain $B(0,R)$.
\end{enumerate}
{\bf Point (1)} We first note that $h$ can be equivalently written as
$$h(\theta,y)=\log \inparen{\sum_{\alpha \in \cF} p_\alpha e^{\inangle{\alpha,y}}} - \inangle{\theta, y}.$$
\noindent 
Thus $h(\theta,y)=\log g_p(e^y) - \inangle{\theta,y}$, where $e^y=(e^{y_1}, e^{y_2}, \ldots, e^{y_m})$ and consequently $h(\theta,y)$ can be evaluated using just the evaluation oracle for $g_p$. For the case of gradients observe first that
$$\nabla_y h(\theta,y) = \frac{1}{g_p(e^y)}\sum_{\alpha \in \cF}\alpha  p_\alpha e^{\inangle{\alpha,y}} - \theta.$$
Since computing $g_p(e^y)$ is easy, it remains to deal with $\sum_{\alpha \in \cF}\alpha  p_\alpha e^{\inangle{\alpha,y}}$.
For this note that the $i$th coordinate of the above is 
$$\sum_{\alpha \in \cF}\alpha_i  p_\alpha e^{\inangle{\alpha,y}} = \frac{d}{dt} g_p(e^{y_1}, \ldots, e^{y_{i-1}},e^{y_i}+t, e^{y_{i+1}}, \ldots, e^{y_m}).$$
The right hand side above is a univariate polynomial $h(t)$ of degree at most $2d$ (since $\cF \subseteq [0,2d]$). 
To compute its derivative it is enough to learn all its coefficients (in fact it is enough to learn the coefficient of $t$ in $h(t)$).
Towards this, note that the evaluation oracle for $g_p$ implies an evaluation oracle for $t\mapsto h(t)$, and hence we can simply evaluate $h$ at $2d+1$ different points and recover its coefficients using polynomial interpolation.
The running time of such a procedure is polynomial in $d$ -- as required.

Note also that, importantly, to implement the first order oracle for $h(\theta, \cdot)$ the oracle $g_p$ is queried only on inputs of polynomial bit complexity, as for every $y\in B(0,R)$ the vector $e^y$ has polynomial bit complexity (since $R=\poly\inparen{m,\log\frac{1}{\delta}}$). Thus the running time of these procedures is polynomial in $m, \log \frac{1}{\delta}$ and $d$.

\noindent {\bf Point (2)} This is not an issue since the domain is already a ball.

\noindent {\bf Point (3)} Note that from the Cauchy-Schwarz inequality, for every $y\in B(0,R)$
$$\inangle{y,\alpha - \theta}\leq \norm{y} \norm{\alpha - \theta} \leq R \cdot d,$$
hence, we have
$$\max_{y\in B(0,R)}h(\theta,y) \leq \log \inparen{e^{L_p}\abs{\cF} \cdot e^{R\cdot d}} \leq L_p \cdot R \cdot d \cdot \log |\cF|.$$
Similarly, for the minimum
$$\min_{y\in B(0,R)}h(\theta,y) \geq - L_p\cdot R \cdot d \cdot \log |\cF|.$$
Clearly the logarithm of the gap: 
$$\log \inparen{\max_{y\in B(0,R)}h(\theta,y) -\min_{y\in B(0,R)}h(\theta,y)}\leq \log (2 L_p\cdot R\cdot d\cdot \log |\cF|)$$
 is polynomially bounded (even without taking the logarithm it is still true), as $\abs{\cF}\leq (2d)^m$.

\noindent {\bf Point (4)} This is clear - given a vector $y_0\in \R^m$, if $\norm{y} \leq R$ we just report $y_0$ to be in the domain and if $\norm{y_0}=R' >R$ then $\{y\in \R^m:\inangle{y,y_0}=\frac{R+R'}{2}\}$ is the required separating hyperplane.

\end{proof}
\section{Applications to TCS}\label{sec:applications}

We present several new and old application of our main result: Theorem~\ref{thm:bound} in TCS.
Since some of them rely on the fact that Newton polytopes of real stable polynomials have low unary facet complexity, below we briefly state the necessary background.

\subsection*{Preliminaries on real stability}\label{ssec:real_stable}
We first define the concept of real stable polynomials that appear in some applications of our results.
For a survey on real stable polynomials we refer the reader to~\cite{Wagner11}.
\begin{definition}
A polynomial $p\in \C[x_1, \ldots, x_m]$ is called real stable if all its coefficients are real numbers and the following condition holds
$$\forall_{z\in \C^m} ~~(\Im(z_i)>0 \mbox{ for all }i=1,2, \ldots, m) ~~\Rightarrow ~~ p(z) \neq 0.$$
\end{definition}
\noindent 
The following fact shows that polytopes arising from such polynomials have low facet complexity.
\begin{fact}[\cite{BC95, Branden07}]\label{fact:rs}
Let $p\in \R[x_1, \ldots, x_m]$ be a real stable polynomial with non-negative coefficients. Then, there exists a ``rank function'' $r: \{-1,0,1\}^m \to \Z$ is, such that the convex hull of $\cF \subseteq \N^m$ -- the support of $p$ can be described as:
$$\conv(\cF)=\{x\in \R^m: ~\forall_{c\in \{-1,0,1\}^m} ~ \inangle{x,c} \leq r(c)\}.$$
\end{fact}
\begin{proof}
The proof is a simple consequence of two results. It was proved in~\cite{Branden07} that the support of a real stable polynomial is a {\it jump system}. Such sets were studied previously in~\cite{BC95}, where a polyhedral characterization, as in the conclusion, was shown.
\end{proof}
\noindent
Another important, yet classical fact about real stable polynomial is that the are log-concave in the positive orthant.
\begin{fact}[\cite{Guler97}]\label{fact:rs_concave}
Let $p\in \R[x_1, \ldots, x_m]$ be a real stable polynomial with non-negative coefficients. Then the function $x\mapsto \log p(x)$ is concave over $x\in \R^{m}_{>0}$.
\end{fact}

\subsection{Bounds for the matrix scaling problem} \label{ssec:matrix_scaling}

Consider the $(r,c)-$matrix scaling problem, where one is given a nonnegative square matrix $A\in \R^{n \times n}$ and two vectors $r,c\in \N^n$ (with $\norm{r}_1=\norm{c}_1$) and the goal is to find a scaling: two positive vectors $x,y\in \R^n_{>0}$ such that for $B$ defined as $B:=XAY$ (with $X=\diag{x}$ and $Y=\diag{y}$) it holds that
$ B\mathbbm{1}=r~~\mbox{and}~~B^\top \mathbbm{1} = c$
where $\mathbbm{1} \in \R^n$ is the all-one vector. In other words, we want the row-sums of the matrix $B$ to be equal to $r$ and column sums to be equal to $c$. 

For applications of matrix scaling, we refer to~\cite{ALOW17,CMTV17}, where recently fast algorithms for matrix scaling were recently derived.
Here let us only mention the problem of approximating the permanent of a nonnegative matrix. 
One can show that 
$\per(XAY) = \per(A) \cdot \prod_{i=1}^n x_i \cdot \prod_{i=1}^n y_i$
 and thus, by scaling (with $r=c=\mathbbm{1}$), one can reduce the problem of computing the permanent of a nonnegative matrix to the problem of computing the permanent of a doubly-stochastic matrix which is better understood.
In particular several useful bounds, such as the Van der Waerden's bound (see~\cite{Gurvits06}), and others (\cite{GS14}), are known for the permanent of doubly stochastic matrices.

One can prove that if such a scaling exists even {\it asymptotically} (i.e., a sequence of scalings exists such that they satisfy the scaling condition in the limit), then it can be recovered from the optimal solution to the following convex program 
\begin{equation}\label{eq:scaling}
\inf_{z\in \R^n} \sum_{i=1}^n r_i \log \inparen{\sum_{j=1}^n A_{i,j}e^{z_j}} - \inangle{c,z};
\end{equation}
\noindent 
see~\cite{Gurvits06,KLRS08,ALOW17,CMTV17}.
Indeed, the scaling is recovered as
$
x_i:=r_i \cdot \inparen{\sum_{l=1}^n A_{i,l} e^{z^\star_l} }^{-1}$ for $i=1,2, \ldots n$ and $y_j:=e^{z^\star_j}$ for $i=1,2, \ldots,n$,
where $z^\star$ stands for the optimal (or approximately optimal) solution to~\eqref{eq:scaling}. The question which arises naturally is: does the optimal (or approximately optimal) scaling $(x,y)$ have polynomial bit complexity? 
Or in other words: can we prove that the vector $z^\star$ has polynomially bounded entries: $\norm{z}_{\infty}=\poly(n, L_A)$? Here $L_A$ denotes the bit complexity of the matrix $A$, i.e., $L_A:= \max_{i,j\in [n]} |\log A_{i,j}|$. 

We interpret the optimization problem~\eqref{eq:scaling} as an instance of entropy maximization which will allow us to apply Theorem~\ref{thm:bound_informal} and to deduce appropriate bounds.
To this end we rewrite the problem~\eqref{eq:scaling} as
$
\inf_{z\in \R^n}  \log \prod_{i=1}^n \inparen{\sum_{j=1}^n A_{i,j}e^{z_j}}^{r_i} - \inangle{c,z}.
$
Hence,  there is a polynomial $p\in \R_{\geq 0}[x_1, x_2, \ldots, x_n]$ such that the above optimization problem is equivalent to
$$
\inf_{z\in \R^n}  \log \frac{p(e^{z_1}, e^{z_2}, \ldots, e^{z_n})}{e^{\inangle{z,c}}}=\inf_{z\in \R^n} \log \inparen{\sum_{\alpha} p_{\alpha} e^{\inangle{\alpha-c,z}}},
$$
where the summation is over the support of $p$, which in this case (if, say, all the entries $A_{i,j}$ are positive) is equal to 
$\cF = \inbraces{\alpha \in \N^n: \sum_{j=1}^n \alpha_j = \norm{r}_1}.$
Hence~\eqref{eq:scaling} is essentially a generalized max-entropy program over the set $\cF$ with expectation  $c\in \conv(\cF)$. 
This allows us to analyze the scaling problem by applying Theorem~\ref{thm:bound} (a generalization of Theorem~\ref{thm:bound_informal}) and prove polynomial bounds on the bit complexity.

Below we use $L_A$ for the bit complexity of the matrix $A$, i.e., $L_A:=\max\{\log |A_{i,j}|: i,j\in \{1,2,\ldots, n\}, ~A_{i,j}\neq 0\}$.

\begin{corollary}
Let $A\in \R^{n\times n}$ be a nonnegative matrix which is asymptotically  $(r,c)-$scalable with $\norm{r}_1=\norm{c}_1=h$. Then for every $\eps>0$ there exists a scaling $(x,y)$ such that:
\begin{align*}
\abs{\sum_{k=1}^n x_i y_k A_{i,k} - r_i}< \eps & ~~~~\mbox{for all } i\in [n],\\
\abs{\sum_{l=1}^n x_l y_j A_{l,j} - c_j}< \eps & ~~~~\mbox{for all } j\in [n],
\end{align*}
and the bit complexities of all entries of $x$ and $y$, i.e., $\max_{i\in [n]} |\log x_i|$ and $\max_{i\in [n]} |\log y_i|$ are bounded by $\poly\inparen{n, L_A, h, \log \frac{1}{\eps}}$. 
\end{corollary}

\begin{proof}
Our strategy is to derive the result from Corollary~\ref{cor:rep}. We first rewrite the objective~\eqref{eq:scaling} so that it matches the form as in Theorem~\ref{thm:bound}. 

\begin{align*}
\sum_{i=1}^n r_i \log \inparen{\sum_{j=1}^n A_{i,j}e^{z_j}} - \inangle{c,z}&= \log \inparen{\prod_{i=1}^n\inparen{\sum_{j=1}^n A_{i,j}e^{z_j}}^{r_i}} - \inangle{c,z}\\
&=\log \inparen{\sum_{\alpha \in \cF} p_\alpha e^{\inangle{\alpha - c,z}}}\\
& = h(c,z)
\end{align*}
For some set $\cF \subseteq \N^m$ and some family of positive numbers $\{p_\alpha\}_{\alpha \in \cF}$.

The next step is to obtain bounds on the various quantities $m,d,M,L_p$ which appear in Corollary~\eqref{cor:rep}. Clearly we choose $m:=n$. Next observe that $\norm{\alpha}_\infty\leq h$ for $\alpha \in \cF$, hence $d\leq h$ and that $\abs{\log p_\alpha}\leq \poly(h, L_A)$ for all $\alpha \in \cF$, hence $L_p =\poly(h, L_A)$. 

Let us now prove that the polytope $\conv(\cF)$ can be described by inequalities with small integer coefficients. To this end note that $p$ can be naturally treated as a polynomial ($p_\alpha$ is the coefficient of $x^\alpha:=\prod_{i\in [n]} x_i^{\alpha_i}$).
$$p(x) = \sum_{\alpha \in \cF} p_\alpha x^\alpha = \prod_{i=1}^n\inparen{\sum_{j=1}^n A_{i,j}x_j}^{r_i},$$
and the support of $p$ is equal to $\cF$. Since $p$ is a product of linear polynomials with nonnegative coefficients, it is a {\it real stable polynomial} (see preliminaries at the beginning of this section for some background). As shown in the preliminaries (Fact~\ref{fact:rs}), the Newton polytope of a real stable polynomial  can be described by inequalities of the form
$$\inangle{a,x} \leq b$$
where $a\in \{-1, 0, 1\}^n$. Hence we can take $M=O(1)$ in the statement of Corollary~\ref{cor:rep}. 

We now apply Corollary~\ref{cor:rep} to obtain a point $z^\star$ and use it to define a scaling by the following formulas
\begin{equation}\label{eq:scaling_xy}
x_i=r_i \cdot \inparen{\sum_{l=1}^n A_{i,l} e^{z^\star_l} }^{-1}~~~~ \mbox{ and } ~~~~ y_j=e^{z^\star_j},~~~~ \mbox{ for } i,j=1,2, \ldots,n.
\end{equation}
\noindent 
 Note that such a pair $(x,y)$ has bit complexity which is polynomial in $n,L_A, h, \log \frac{1}{\eps}$, hence it remains to reason about the precision of the resulting scaling.

By a direct calculation one obtains that
$$\abs{\sum_{k=1}^n x_i y_k A_{i,k} - r_i}=0 ~~~~\mbox{for all } i\in [n],$$
hence it remains to prove a bound on the precision of the scaling with respect to columns. For this, note that
$$\inparen{\nabla_z h(c,z)}_j = \sum_{l=1}^n x_l y_j A_{l,j} - c_j.$$
Since Corollary~\ref{cor:rep} implies that
$$\norm{\nabla_z h(c,z^\star)}_2<\eps$$
the bound on the scaling precision follows. 
\end{proof}

\subsection{Computability of recent continuous relaxations for counting and optimization Problems}\label{sec:app_cp}
The recent works~\cite{SV17} and~\cite{AO17} study counting and optimization problems involving polynomials and provide approximation algorithms for them. 

In the discussion below, for concreteness, we follow~\cite{SV17}.
 -- the setting of~\cite{AO17} is similar.

\cite{SV17} consider the following approach to construct approximation algorithms for counting and optimization problems: consider a multiaffine polynomial $p \in \R[x_1, x_2, \ldots, x_m]$ with nonnegative coefficients and a family of sets $\cB \subseteq 2^{[m]}$.\footnote{The polynomial in this problem is provided as an evaluation oracle. Similarly the family $\cB$ is given implicitly in the input, as a separation oracle for the convex hull of $\cB$.}
Let us denote by $p_\alpha$  the coefficient of the monomial $x^\alpha:= \prod_{i=1}^m x_i^{\alpha_i}$ in $p$.
Then the problems considered are to compute
\begin{equation}
\label{eq:cnt_max}
p_\cB := \sum_{\alpha \in \cB}p_\alpha ~~~~~~\mbox{and}~~~~~~~p^{\max}_\cB:=\max_{\alpha \in \cB} p_\alpha .
\end{equation}
\noindent
One particular application where one is required to solve such problems is when dealing with constrained Determinantal Point Processes~\cite{KuleszaTaskar12,CDKSV17,dpp_data_summarization}.
There, the polynomial $p(x)$ is of the form $\sum_{S\subseteq [m]} \det(L_{S,S})x^S$, where $L \in \R^{m \times m}$ is a PSD matrix and $L_{S,S}$ denotes the submatrix of $L$ corresponding to rows and columns in $S\subseteq [m]$. 
Solving such counting and optimization problems allows one to draw samples from constrained DPPs, i.e., distributions where the probability of a set $S$ is proportional to $\det(L_{S,S})$ when $S\in \cB$ and is $0$ when $S\notin \cB$. %
Such distributions  have various interesting uses in data summarization and fair and diverse sampling (see~\cite{CDKSV17}).

To tackle the counting problem of computing $p_\cB$ (as in~\eqref{eq:cnt_max}) the following relaxation is considered: 
\begin{equation}
\label{eq:capa}
\capa_\cB(p):= \sup_{\theta \in P(\cB)} \inf_{x>0} \frac{p(x)}{\prod_{i=1}^m x_i^{\theta_i}}.
\end{equation}
The relaxation for maximization is derived as a small modification of the above, hence we focus on $\capa_\cB(p)$  in this discussion.
  $\capa_\cB(p)$ approximates $p_\cB$ provably well under the assumption that (roughly)  the polynomial $p$ is {\it real stable} (see~\cite{Wagner11}) and that the family $\cB$ has a {\it matroid} structure. 
However, the question of efficient computability of these relaxations  was not established in~\cite{SV17}. 
The results of this paper allow us to deduce polynomial time algorithms for this relaxation in a fairly general setting.
We interpret this relaxation as a variant of a max-entropy program which then allows us to apply Theorem~\ref{thm:comp_entro} to reason about its computability. 

After taking the logarithm and replacing the variables $x_i>0$ by $e^{y_i}$ with $y_i\in \R$ for $i=1,2, \ldots, m$ in $\capa_\cB(p)$ we get

$$\log \capa_\cB(p)=\sup_{\theta \in P(\cB)} \inf_{y\in \R^m} \log \inparen{\sum_{\alpha} p_{\alpha} e^{\inangle{y,\alpha - \theta}}},$$
where the summation in the inner optimization problem runs over $\alpha \in \supp(p)$, i.e., the set of all monomials $\alpha\in \N^m$ whose coefficient $p_\alpha$ is non-zero. 
Note that the inner optimization problem is the dual of a generalized max-entropy program, as in~\eqref{eq:gen_max_entropy}.
Hence,  the inner optimization problem searches for a probability distribution over monomials of the polynomial $p$, whose expectation is $\theta$ and has the smallest possible $KL$-distance to the distribution given by the  coefficients of $p$.

By taking into account the outer optimization over $\theta$, one can see that $\log \capa_\cB(p)$ minimizes the $KL$-distance of a distribution over monomials of $p$ restricting its expectation to be in the convex hull of $\cB$.
This makes it suitable to apply Theorem~\ref{thm:comp_entro} and deduce polynomial time computability. 
Finally, we note that having a max-entropy solver for all $\theta$ is crucial here. 
If $\cF\subseteq \N^m$ denotes the support of $p$, then for any $\theta \notin \conv(\cF)$ the value of the entropy maximization problem is $-\infty$. 
Hence, in the outer optimization problem, the variable $\theta$ is  constrained to be in $P(\cB) \cap \conv(\cF)$. 
Thus, the optimal solution $\theta^\star$ of such an optimization problem might lie at the boundary of $\conv(\cF)$ or very close to it, hence the result of \cite{SinghV14} does not yield polynomial bit complexity bounds in this setting. We prove the following
\begin{theorem}\label{thm:compsv}
Let $p\in \R[x_1, x_2, \ldots, x_m]$ be a polynomial with nonnegative coefficients with support $\cF \subseteq \N^m$ and let $\cB \subseteq \N^m$. Assume that the unary facet complexity of $\conv(\cF)$ is $M$ and denote $d=\max(\diam(\cF), \diam(\cB))$. There is an algorithm which given an evaluation oracle for $p$, a separation oracle for $\conv(\cF)$, a separation oracle for $\conv(\cB)$ and an $\eps>0$ computes a number $X$ such that $1-\eps<\frac{X}{\capa_\cB(p)}<1+\eps$ in time $\poly(m, L_p, \log d, M, \log \frac{1}{\eps})$.
\end{theorem}
\noindent 

\begin{proof}
Denote by $P \subseteq \R^m$ the convex hull of $\cF$ -- the support of $p$.
As discussed above, $\capa_\cB(p)$ can we rewritten equivalently as
$$\log \capa_\cB(p) = \sup_{\theta \in P(\cB)} \inf_{y\in \R^m} \log \inparen{\sum_{\alpha \in \cF} p_{\alpha} e^{\inangle{\alpha-\theta, y}}}=\sup_{\theta \in P(\cB)} \inf_{y\in \R^m}h(\theta,y),$$
where $\cF \subseteq \N^m$ denotes the support of $p$. Thus we obtain a form of $\log \capa_\cB(p)$ in terms of a function $h$ as in~\eqref{eq:h}. Let us denote
$$g(\theta) = \inf_{y\in \R^m} \log \inparen{\sum_{\alpha \in \cF} p_{\alpha} e^{\inangle{\alpha-\theta, y}}}.$$
Hence the goal is to solve $\sup_{\theta \in P(\cB)} g(\theta)$. In fact, since $g(\theta)=-\infty$ whenever $\theta \notin P$, we can rewrite it equivalently as
$$\sup_{\theta \in P(\cB) \cap P} g(\theta).$$
Importantly, since we are only interested in an additive $\eps$-approximation of the above quantity we can apply Theorem~\ref{thm:bound} to replace $g$ by the following Lipscthiz proxy:
 $$\wt{g}(\theta):=\inf_{y\in B(0,R)}h(\theta,y)$$ for an appropriate number $R$ -- polynomial in the input size and $\log \frac{1}{\eps}$ (as in Theorem~\ref{thm:bound}). 

We now apply the ellipsoid method to find an additive $\eps$-approximation to 
\begin{equation}\label{eq:simp_wtg}
\sup_{\theta \in P(\cB) \cap P} \wt{g}(\theta).
\end{equation}
Firstly, observe that the function $\wt{g}$ is concave -- as a pointwise infimum of affine functions.
Now, following the treatment of the ellipsoid method in~\cite{BentalN12} (Theorem 8.2.1) we have to address the following requirements to obtain polynomial running time
\begin{enumerate}
\item Construct a first order oracle for $\theta \mapsto \wt{g}(\theta)$, i.e., an efficient way to evaluate values $\wt{g}(\theta)$ and (sub)gradients of this function for $\theta \in P(\cB) \cap P$.
\item Provide an outer ball -- containing the domain of the considered optimization problem and an inner ball -- contained in the domain. The radii of them should be of polynomial bit complexity.
\item A bound on the gap between the maximum and minimum value of $\wt{g}(\theta)$ in $P(\cB) \cap P$. More precisely, defining $D:=\max_{\theta \in P(\cB) \cap P} \wt{g}(\theta) -\min_{\theta \in P(\cB) \cap P} \wt{g}(\theta)$ we would like $\log D$ to be polynomially bounded.

\item Provide a separation oracle for the domain $P(\cB) \cap P$.
\end{enumerate}

\parag{Point (1).}
To obtain an evaluation oracle for $\wt{g}$ note that Theorem~\ref{thm:comp_entro} gives an algorithm to compute $\wt{g}$ up to $\delta$ additive error in time polynomial in $\log \frac{1}{\delta}$. 
This provides an approximate (weak) evaluation oracle, which is still enough to run the ellipsoid method (using shallow cuts, see~\cite{Grotschel1988}).

Let us now discuss the oracle for the gradient of $\wt{g}$.
It is a standard fact in convex programming that for any $\theta \in P$ it holds that the point
$$y^\star:=\argmin_{y\in B(0,R)}  h(\theta,y)$$
is a subgradient of $\wt{g}$ at $\theta$. 
Using Theorem~\ref{thm:comp_entro} we can efficiently compute a vector $y$ which is a $\delta$-approximation to $y^\star$ in the sense that $h(\theta,y)-h(\theta,y^\star)\leq \delta$.
We use this algorithm to implement an approximate subgradient oracle for $\wt{g}$, i.e., for a given $\theta$ we output the vector $y$ as above, using the algorithm in Theorem~\ref{thm:comp_entro}. 

It remains to justify why is such an approximate gradient enough to run the ellipsoid method. 
Note that if $\theta \in P$ and  $y\in B(0,R)$ is such that $h(\theta,y)\leq \wt{g}(\theta) +\delta$ then for any $\theta'\in P$ we have

$$g(\theta') \leq h(\theta',y) = h(\theta,y) + \inangle{\theta'-\theta,y}\leq g(\theta)+\delta + \inangle{\theta'-\theta,y},$$
where the equality above holds because $h$ is an affine function of $\theta$. Hence, by applying such a $y$ as a gradient oracle, we obtain a separation hyperplane which never cuts out a point $\theta'\in P$ of value $g(\theta')>g(\theta)+\delta$.
This, along with the fact that $y$ can be found in time polynomial in $\log \frac{1}{\delta}$, allows us to apply the shallow cut ellipsoid method.
We also refer the reader to the full version of~\cite{SinghV14} and the proof of Theorem 2.11 therein where a detailed discussion of an ellipsoid algorithm for a related problem is provided.

\parag{Point (2).}
An outer ball is easy to obtain since $P(\cB) \subseteq [0,1]^m$.
For the inner ball, suppose first that $P(\cB) \cap P$ is full-dimensional. 
Then, since the vertices of both these polytopes have small integer entries, one can show using standard techniques that a ball of radius $\Omega(e^{-\poly(m)})$ can be fit inside this polytope.
In the non-full-dimensional case one has to work in a linear subspace $H\subseteq \R^m$ on which $P(\cB) \cap P$ is full-dimensional; $H$ can be found given separation oracles of $P(\cB)$ and $P$ using standard subspace identification techniques (see e.g.~\cite{Grotschel1988}).

\parag{Point (3).}
As in the proof of Theorem~\ref{thm:comp_entro} we observe that for every $y\in B(0,R)$
$$-L_p \cdot R\cdot m \cdot d \log d\leq \log  \inparen{\sum_{\alpha \in \cF} p_{\alpha} e^{\inangle{\alpha-\theta, y}}}\leq L_p \cdot R\cdot m \cdot d \log d, $$
\noindent 
and hence the gap (and hence clearly its logarithm as well) is polynomially bounded.

\parag{Point (4).}
This is clear as we are given separation oracles for both $P(\cB)$ and $P$.

\end{proof}

\subsection{Proof of computability of worst-case Brascamp-Lieb constants}\label{sec:app_bl}

\noindent 
\begin{proofof}{of Theorem~\ref{thm:bl_comp}}
Let us denote by $v_j \in \R^n$ the only row of the matrix $B_j$ and by $V\in \R^{m \times n}$ the matrix collecting all the $v_j$'s as rows. Then the the Brascamp-Lieb constant (for the rank-$1$ case) can be computed as

\begin{equation}
\bl(B,p)=\inf_{x>0} \frac{\det\inparen{\sum_{j=1}^m p_j x_j v_j v_j^\top}}{\prod_{j=1}^m x_j^{p_j}}.
\end{equation}
Note that the numerator is simply a polynomial $r(x)=\sum_{S\subseteq [m], |S|=n} p^S x^S \det(V_S V_S^\top)$, where $V_S$ is the submatrix of $V$ corresponding to the index set $S\subseteq [m]$ and $x^S:=\prod_{i\in S} x_i$.
Therefore, computability of $\bl(B,p)$ directly follows from our general result (Theorem~\ref{thm:comp_entro}) on computability of maximum entropy distributions.

The problem of computing the (logarithm of the) worst-case constant can be reformulated as

$$\sup_{p\in P_B} \inf_{y\in \R^m} \log \inparen{ \sum_{\alpha \in \cF}r_\alpha(p) e^{\inangle{\alpha - p, y}}},$$
\noindent 
where $\cF \subseteq \{0,1\}^m$ is the support of $r$ and $r_\alpha(p)$ is the coefficient of $\alpha$ in r, as a function of $p\in P_B$. 
Thus, the inner optimization problem is the dual of a certain max-entropy program. 
Moreover, the function
$$p\mapsto  \log \inparen{ \sum_{\alpha \in \cF}r_\alpha(p) e^{\inangle{\alpha - p, y}}}=\log \inparen{\sum_{S\subseteq [m], |S|=n} p^S x^S \det(V_S V_S^\top)}$$ is concave over $\R^{m}_{>0}$, as the polynomial $p \mapsto \sum_{S\subseteq [m], |S|=n} p^S x^S \det(V_S V_S^\top)$ (treating $x>0$ as a vector of positive constants) is real stable, see Fact~\ref{fact:rs_concave}.
From this point on, the argument follows along the same lines as the proof of Theorem~\ref{thm:compsv}.

\end{proofof}

\bibliographystyle{plain}
\bibliography{references} 

\appendix

\section{Example when the empirical mean lies on the boundary}\label{eq:sec_boundary}
In this section, we show that for certain well-behaved classes of polytopes the empirical mean of a sequence of uniform samples will yield a $\theta$ that is very likely to lie on the boundary, and hence the result of~\cite{SinghV14} is not applicable in such cases.
\begin{theorem}[Empirical mean arbitrarily close to the boundary]
There exists a family of vectors $\cF \subseteq \{0,1\}^m$ such that $P:=\conv(\cF)$ contains a ball of radius $\Omega\inparen{1}$, has unary facet complexity $1$, and if $X_1, X_2, X_3, \ldots \in \cF$ are independent uniform samples from $\cF$, then for every $N\leq 2^m$ it holds
$$\prob{\frac{1}{N}\sum_{i=1}^N X_i \notin \partial P}= O\inparen{\frac{N}{2^m}},$$
where $\partial P$ denotes the boundary of the polytope $P$.
\end{theorem}
\begin{proof}
Let $\cF$  be
$$\cF :=\inparen{\{0\}\times \{0,1\}^{m-1}} \cup \{(1,0,0,\ldots, 0)\}.$$
The first step is now to observe that the polytope $P:=\conv(\cF)$ can be simply described by:
$$P:=\{x\in \R^m: 0\leq x_1\leq 1,~~0\leq x_i \leq 1-x_1, ~~\mbox{for $i=1,2,\ldots, m$}\}.$$
From this description, it already follows that $P$ has unary facet complexity of $1$, because all coefficients used to describe this polytope have magnitude $1$. Moreover, another simple consequence of the above description of $P$ is that
$$[0,1/2]^m \subseteq P.$$
In particular, it follows that $P$ contains a Euclidean ball of radius $\Omega(1)$ (for instance the ball centered at $(1/4,1/4, \ldots, 1/4)$ with radius $1/4$).

It remains to establish a bound on the probability that the empirical mean of $N\leq 2^m$ independent samples lies on the boundary of $P$. For this, observe first that 
$$F=\{0\}\times [0,1]^{m-1}$$
is an $(m-1)$-dimensional facet of $P$ and in particular, is a subset of the boundary. Moreover
$$\prob{X_i \in F} = \frac{1}{|\cF|} = \frac{1}{2^{m-1}+1}.$$
Consequently, we have
$$\prob{\sum_{i=1}^N X_i \in \partial P} \geq \prob{\frac{1}{m}\sum_{i=1}^N X_i \in F} =\prod_{i=1}^N \prob{X_i\in F} = \inparen{1-\frac{1}{2^{m-1}+1}}^N .$$
Thus, 
$$\prob{\sum_{i=1}^N X_i \notin \partial P}\leq 1 -  \inparen{1-\frac{1}{2^{m-1}+1}}^N= O\inparen{\frac{N}{2^m}},$$
the latter follows because $\frac{N}{2^m}=O(1)$.
\end{proof}

\end{document}